\documentclass[11pt,a4paper]{article}
\makeatletter
\RequirePackage[utf8]{inputenc}
\RequirePackage[T1]{fontenc}
\RequirePackage[english]{babel}
\RequirePackage{lmodern,microtype}
\RequirePackage[a4paper,top=2.5cm,bottom=2.5cm,left=2.7cm,right=2.7cm,footskip=1.2cm]{geometry}
\RequirePackage{amsmath,amssymb,amsthm,mathtools,bm}
\RequirePackage{graphicx,booktabs,array,tabularx,multirow,makecell,longtable}
\RequirePackage{caption,subcaption,xcolor,enumitem,csquotes,tikz,siunitx}
\RequirePackage{needspace}
\usetikzlibrary{arrows.meta,positioning,shapes.geometric}
\RequirePackage[backend=biber,style=authoryear-comp,natbib=true,url=true,doi=true,isbn=false,maxbibnames=99,maxcitenames=2,uniquelist=false,sorting=nyt]{biblatex}
\RequirePackage[colorlinks=true,linkcolor=black,citecolor=black,urlcolor=black]{hyperref}
\RequirePackage[noabbrev,capitalize]{cleveref}
\newcolumntype{L}[1]{>{\raggedright\arraybackslash}p{#1}}
\newcolumntype{Y}{>{\raggedright\arraybackslash}X}
\setlist{nosep,leftmargin=*}
\let\ffraStandardSection\section
\renewcommand{\section}{\Needspace{7\baselineskip}\ffraStandardSection}
\theoremstyle{plain}
\newtheorem{proposition}{Proposition}
\newtheorem{theorem}[proposition]{Theorem}
\newtheorem{lemma}[proposition]{Lemma}
\newtheorem{corollary}[proposition]{Corollary}
\theoremstyle{definition}
\newtheorem{definition}{Definition}
\newtheorem{assumption}{Assumption}

\theoremstyle{remark}
\newcommand{\ffraseries}{ForesightFlow Research}
\newcommand{\ffraauthor}{Maksym Nechepurenko}
\newcommand{\ffraaffiliation}{Research Department of Devnull FZCO, Dubai, UAE}
\newcommand{\ffraemail}{maksym@devnull.ae}
\newcommand{\ffraorcid}{0000-0002-9515-8841}
\newcommand{\ffraversion}{r0.0.0}
\newcommand{\ffradate}{\today}

\newcommand{\ffraSetup}[1]{\setkeys{ffra}{#1}}
\RequirePackage{keyval}
\define@key{ffra}{series}{\renewcommand{\ffraseries}{#1}}
\define@key{ffra}{author}{\renewcommand{\ffraauthor}{#1}}
\define@key{ffra}{affiliation}{\renewcommand{\ffraaffiliation}{#1}}
\define@key{ffra}{email}{\renewcommand{\ffraemail}{#1}}
\define@key{ffra}{orcid}{\renewcommand{\ffraorcid}{#1}}
\define@key{ffra}{version}{\renewcommand{\ffraversion}{#1}}
\define@key{ffra}{date}{\renewcommand{\ffradate}{#1}}
\define@key{ffra}{status}{\renewcommand{\ffrastatus}{#1}}
\renewcommand{\maketitle}{\begin{center}\small\textsc{\ffraseries}\par\vspace{0.75em}{\LARGE\bfseries \@title\par}\vspace{0.8em}{\large \ffraauthor\par}\vspace{0.25em}{\small \ffraaffiliation\par\texttt{\ffraemail}\quad ORCID \ffraorcid\par\vspace{0.25em}\ffradate\quad\textbar\quad Version \ffraversion\par}\end{center}\vspace{0.5em}}
\newcommand{\ffraKeywords}[1]{\noindent\textbf{Keywords:} #1\par}
\newcommand{\ffraJEL}[1]{\noindent\textbf{JEL:} #1\par}
\newcommand{\ffraEvidence}[1]{\PackageWarning{foresightflow_research_article}{\string\ffraEvidence is deprecated and renders no text in FFRA-v1.0.3}}

\definecolor{seriesblue}{RGB}{31,78,121}
\definecolor{seriesgray}{RGB}{242,244,247}
\definecolor{serieslight}{RGB}{248,249,251}
\definecolor{seriesdark}{RGB}{52,58,64}
\definecolor{seriesorange}{RGB}{180,96,33}
\usetikzlibrary{fit,calc,decorations.pathreplacing,matrix,backgrounds,patterns}
\tikzset{seriesbox/.style={draw=black!70,rounded corners=2pt,align=center,inner sep=5pt,minimum height=9mm,fill=seriesgray},seriesarrow/.style={-{Latex[length=2.2mm]},thick,draw=black!75}}
\makeatother
\newcommand{\ObsExact}{\mathsf{E}}
\newcommand{\ObsInterval}{\mathsf{I}}
\newcommand{\ObsSnapshot}{\mathsf{S}}
\newcommand{\ObsProxy}{\mathsf{Prx}}
\newcommand{\ObsUnmeasured}{\mathsf{U}}
\newcommand{\ObsConflict}{\mathsf{X}}

\newtheorem{hypothesis}{Hypothesis}

\newtheorem{designprinciple}{Design Principle}
\crefname{designprinciple}{Design Principle}{Design Principles}
\Crefname{designprinciple}{Design Principle}{Design Principles}
\crefname{empiricalfinding}{Empirical Finding}{Empirical Findings}
\Crefname{empiricalfinding}{Empirical Finding}{Empirical Findings}
\usepackage{pgfplots}
\pgfplotsset{compat=1.18}
\usepackage{adjustbox}
\usepackage{etoolbox}
\usepackage{seqsplit}
\hypersetup{
  pdftitle={Resolution Is Not Settlement, Part II: Protocol Finality and Observed Redemption on Polymarket},
  pdfauthor={Maksym Nechepurenko},
  pdfsubject={Prediction markets; protocol finality; Conditional Tokens; observed redemption; Polymarket},
  pdfkeywords={prediction markets, Polymarket, Conditional Tokens, protocol finality, payout vectors, redemption, survival analysis}
}

\newcommand{\E}{\mathbb E}

\newcommand{\ind}[1]{\mathbf 1\!\left[#1\right]}

\newcommand{\code}[1]{\texttt{#1}}
\newcommand{\eventcode}[1]{\texttt{\seqsplit{#1}}}

\newcommand{\tprepare}{t_{\mathrm{prepare}}}
\newcommand{\toracle}{t_{\mathrm{oracle\text{-}final}}}
\newcommand{\tconsume}{t_{\mathrm{adapter\text{-}consume}}}
\newcommand{\tadapter}{t_{\mathrm{adapter\text{-}terminal}}}
\newcommand{\tresolve}{t_{\mathrm{protocol\text{-}final}}}
\newcommand{\tredeemable}{t_{\mathrm{redeemable}}}
\newcommand{\tredeem}{t_{\mathrm{first\text{-}redeem}}}
\newcommand{\tredeempos}{t_{\mathrm{first\text{-}positive\text{-}redeem}}}
\newcommand{\tstar}{T^{\star}}

\newcommand{\Yvec}{\mathbf Y}
\newcommand{\pivec}{\boldsymbol\pi}

\newcommand{\Cohort}{\mathcal Q^{\star}}
\newcommand{\Conditions}{\mathcal C}
\newcommand{\Redemptions}{\mathcal R}
\newcommand{\km}{\widehat S_{\mathrm{KM}}}

\newcommand{\occupancy}{\Omega}
\newcommand{\sourceobs}{\mathcal O}
\newcommand{\canon}{\mathcal E}
\newcommand{\valuegate}{\mathrm{G\text{-}VALUE}}

\setlist[description]{style=nextline,leftmargin=0pt,labelindent=0pt}

\crefname{designprinciple}{Design Principle}{Design Principles}
\Crefname{designprinciple}{Design Principle}{Design Principles}
\crefname{empiricalfinding}{Empirical Finding}{Empirical Findings}
\Crefname{empiricalfinding}{Empirical Finding}{Empirical Findings}

\ffraSetup{series={ForesightFlow Research \textperiodcentered{} Event-Linked Perpetuals \textperiodcentered{} Paper 5, Part II},author={Maksym Nechepurenko},affiliation={Research Department of Devnull FZCO, Dubai, UAE},email={maksym@devnull.ae},orcid={0000-0002-9515-8841},version={r0.6.1},date={August 2026},status={}}

\title{Resolution Is Not Settlement, Part II:\\
Protocol Finality and Observed Redemption on Polymarket}
\author{\SeriesAuthor\thanks{\SeriesAffiliation}}
\date{\SeriesDate}

\begin{document}
\maketitle
\begin{abstract}
An Oracle result is not yet a protocol payout, a redeemable position is not yet collateral in a holder's account, and a redemption event is not a complete measure of economic entitlement. This companion paper develops an event-sourced framework for Polymarket conditions from preparation through protocol finality and observed holder realization.

The empirical design uses three Conditional Tokens Framework event families derived from a pinned contract application binary interface (ABI): \eventcode{ConditionPreparation}, \eventcode{ConditionResolution}, and \eventcode{PayoutRedemption}. It separates the contract-wide acquisition universe from the frozen Polymarket adapter-question cohort. The exact bridge contains 108,638 linked conditions. Of these, 99,283 have an observed protocol-resolution event by the fixed snapshot at Polygon block 90,114,204 (2026-07-12T17:11:41Z).

Among resolved exact-linked conditions, 92,158 have an observed redemption of any amount and 91,817 have an observed positive-payout redemption; condition-specific Kaplan--Meier medians from first protocol resolution are 182 and 200 seconds respectively. The exact-linked payout taxonomy contains 53,847 canonical $(0,1)$ vectors, 45,024 canonical $(1,0)$ vectors, 410 fifty-fifty vectors, two other valid vectors, and 9,355 conditions with no observed resolution. Cross-contract Oracle--adapter--protocol ordering is reported conservatively: 823 conditions have an interval-qualified terminal generation, 48 have multiple candidate generations, 91,638 have no compatible terminal generation in the frozen evidence, and 16,129 are right-censored or otherwise unevaluable.

The formal results show that Oracle finality does not identify protocol finality, protocol finality does not identify holder realization, and redemption events alone do not identify the fraction of entitlement redeemed without an independent balance-consistent entitlement denominator. The paper therefore reports observed event timing and payout realization without claiming G-VALUE, outstanding winning-token supply, or unredeemed collateral value.
\end{abstract}

\ffraKeywords{prediction markets; event-linked perpetual futures; finality; Polymarket.}
\ffraJEL{G13, G14, G18.}

\section{Introduction}
\label{sec:introduction}

A prediction market can have a final oracle result while its protocol payout vector remains unrecorded.  A payout vector can be recorded while no holder has redeemed.  A holder can be technically able to redeem and nevertheless defer or never exercise that action.  Treating these states as one timestamp conflates adjudication, protocol accounting, asset availability, and holder behavior.

The distinction is economically material whenever an event claim is financed, margined, converted, or used as collateral.  A lender cannot regard an oracle result as cash if the adapter has not consumed it, the Conditional Tokens Framework (CTF) has not recorded payouts, or the relevant claim has not been converted into collateral under an enforceable rule.  Conversely, delayed holder action after protocol finality need not indicate a protocol defect.  It can reflect transaction costs, batching, custody operations, gas conditions, zero-payout cleanup, portfolio policy, or simple inattention.  The same observed delay can therefore belong to different balance sheets and different actors.

Part~I reconstructs Polymarket's semantic and oracle layer through the adapter-terminal transition \citep{nechepurenko2026_p5a}.  It shows that technical request creation, proposal, dispute, reset, oracle settlement, and adapter consumption are distinct states.  Part~II begins at the protocol layer.  It asks when a condition is registered, when its ordered payout vector becomes the contract's redemption rule, when the first holder action appears, and which value claims remain unidentified without a complete entitlement denominator.

The downstream acquisition layer spans the pinned Conditional Tokens contract over a fixed block interval, while the scientific population remains the frozen adapter-question cohort inherited from Part~I.  This distinction is load-bearing: a contract-wide event total is a coverage and transport result, not the denominator for a Polymarket cohort rate.  Every population estimate in Part~II therefore enters only after exact protocol identity has linked an event back to the frozen cohort.

\subsection{Why protocol finality deserves a separate paper}

A single phrase such as ``the market resolved'' can refer to at least four different events:

\begin{enumerate}
  \item the oracle answer is no longer contestable through the ordinary path;
  \item the adapter has consumed that answer and reached a terminal state;
  \item the conditional-token contract has recorded payout numerators;
  \item a holder has redeemed a position into collateral.
\end{enumerate}

These events can share one block timestamp, occur in different transactions, or be separated by a long interval.  Their ordering also depends on the emitted logs.  In a standard downstream call, the adapter can invoke the protocol payout function before emitting its own terminal event.  It is therefore incorrect to force the companion papers into a rigid timestamp sequence in which ``adapter finality'' always precedes ``protocol finality.''  The correct object is an exact event graph with block, transaction, and log-order keys.

The protocol layer creates a further identification problem that does not arise at the oracle layer.  Redemption logs reveal actions that occurred.  They do not directly reveal the total quantity of winning positions that could have been redeemed.  Without a balance-consistent entitlement denominator, the event ledger supports an observed payout curve but not the fraction of all entitlement completed.  This distinction is especially important when a leveraged design attempts to use delayed redemption as a measure of protocol risk or capital lock-up.

\subsection{Companion-paper boundary}

The two parts use one frozen adapter-question cohort but different primary units.  Part~I uses the adapter-question and oracle request generation.  Part~II uses the condition, condition-resolution event, and redemption event.  The boundary is analytical rather than an assertion that the adapter-terminal log must precede the protocol-resolution log.  Both can arise within one downstream call, and their observed order is determined by canonical log indices.  \Cref{fig:p2_layers} fixes this division of responsibility.

\begin{figure}[htbp]
\centering
\resizebox{0.95\textwidth}{!}{%
\begin{tikzpicture}[
  node distance=8mm and 11mm,
  every node/.style={font=\footnotesize},
  layer/.style={seriesbox,minimum width=2.45cm,minimum height=1.0cm},
  registry/.style={seriesbox,dashed,minimum width=2.45cm,minimum height=1.0cm,fill=white},
  arr/.style={seriesarrow}
]
\node[font=\scriptsize\bfseries,seriesblue] (p1lab) at (0,1.15) {Part I evidence lane};
\node[layer,right=10mm of p1lab] (oracle) {Oracle finality};
\node[layer,right=14mm of oracle] (adapter) {Adapter-terminal\\log};

\node[font=\scriptsize\bfseries,seriesblue] (p2lab) at (0,-1.15) {Part II evidence lane};
\node[registry,right=10mm of p2lab] (prepare) {Condition\\prepared};
\node[layer,right=14mm of prepare] (protocol) {Protocol payout\\recorded};
\node[layer,right=of protocol] (redeemable) {Claim\\redeemable};
\node[layer,right=of redeemable] (redeem) {Observed holder\\redemption};

\draw[arr] (oracle) -- (adapter);
\draw[arr] (prepare) -- (protocol);
\draw[arr] (protocol) -- (redeemable);
\draw[arr] (redeemable) -- (redeem);
\draw[arr,dashed] (oracle.south east) -- (protocol.north west);
\draw[<->,thick,dashed,seriesblue] (adapter.south) -- (protocol.north);
\node[font=\scriptsize,align=center] at ($(adapter.east)!0.5!(protocol.east)+(12mm,0)$) {same transaction possible;\\actual log order measured};
\node[font=\scriptsize,align=center] at ($(oracle.south east)!0.5!(protocol.north west)+(-4mm,-1mm)$) {downstream\\consumption};
\end{tikzpicture}%
}
\caption{Companion-paper boundary.  Part I terminates at the observed adapter-terminal state; Part II studies protocol payout recording and holder action.  The division is analytical rather than a universal temporal ordering: protocol-resolution and adapter-terminal logs can occur in one downstream transaction, and canonical log indices determine their observed order.}
\label{fig:p2_layers}
\end{figure}
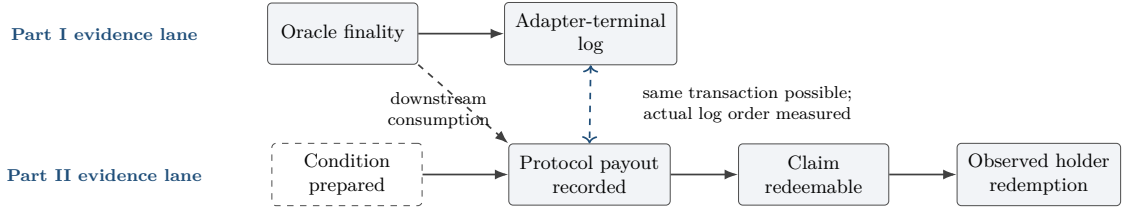

Condition preparation belongs to the protocol architecture but is not a finality event.  It can occur before the external event, before an oracle request becomes answerable, and before any terminal information exists.  Protocol finality begins when the payout rule is recorded.  Redeemability follows from that rule and the deployed contract state.  Redemption is later still: it is an observed holder action.

\subsection{Research questions}

The paper asks seven questions.

\begin{enumerate}
  \item How are frozen adapter-questions mapped exactly to CTF conditions, and where does attrition arise?
  \item Which event and contract identities establish preparation, protocol finality, and redemption without relying on historical stream labels or mutable application-programming-interface snapshots?
  \item How do oracle finality, adapter terminality, and \code{ConditionResolution} relate: in the same transaction, in which log order, in the same block, or later?
  \item Which payout vectors are recorded, including canonical YES/NO outcomes, neutral outcomes, and other valid or malformed vectors?
  \item How long after protocol finality does the first observed redemption occur, and how much right censoring remains at the fixed follow-up snapshot?
  \item How does first positive-payout redemption differ from any redemption, including zero-payout cleanup before a later positive event?
  \item Which quantities are identified by protocol events and which require a separate balance, supply, or entitlement reconstruction?
\end{enumerate}

\subsection{Contributions}

The paper makes seven contributions.

First, it provides an exact event-sourced representation of protocol finality that is separate from oracle adjudication.  Event families are selected from application-binary-interface (ABI)-derived topic hashes and canonical chain identities rather than mutable stream labels or terminal market snapshots.

Second, it defines an exact mapping and attrition design from adapter-question to condition, condition resolution, and redemption.  Title similarity, fuzzy text matching, and nearest-time fallback are excluded from canonical relations.

Third, it develops a payout-vector taxonomy that preserves ordered integer numerators and separates raw protocol representation from normalized economic payout class.

Fourth, it decomposes downstream timing into oracle-to-protocol, oracle-to-adapter, signed adapter--protocol emission order, protocol-to-first-redemption, and protocol-to-first-positive-redemption intervals.  A zero wall-clock gap within one transaction does not erase the exact state order.

Fifth, it specifies right-censored and multistate holder-action outcomes.  A zero-payout redemption ends the ``any redemption'' endpoint but does not preclude a later positive-payout redemption by the same or another holder.

Sixth, it restores the economics of the finality gap without equating all observed discounts or delays with oracle risk.  The paper distinguishes protocol inability, operational delay, holder optionality, and a missing value denominator.

Seventh, it establishes a strict observation-versus-entitlement boundary.  Observed redemption events identify realized actions and paid collateral, but not the fraction of all redeemable entitlement completed unless an independent entitlement denominator is reconstructed.

\subsection{Evidence discipline}

The paper distinguishes five evidence layers:

\begin{enumerate}
  \item exact contract events and canonical chain identities;
  \item deterministic linkage and event-sourced reconstruction;
  \item protocol state and ordered payout rules;
  \item right-censored holder-action outcomes;
  \item unmeasured value denominators and off-chain economic identities.
\end{enumerate}

A large event count does not upgrade incomplete range coverage.  A protocol-resolution event does not prove holder cash realization.  A redemption address is not a beneficial owner.  An observed cumulative payout amount is not divided by an unobserved entitlement denominator.  Every rate, hazard, and duration reports its applicable population or risk set.

\subsection{Position in the wider programme}

Part~II supplies the downstream state variables required by three subsequent papers.  Paper~6 uses condition identity, protocol finality, and delivery feasibility to analyze state-triggered margin and perpetual-to-conditional-token conversion.  Paper~8 uses the distinction between mechanism finality, venue/protocol finality, and cash realization in its functional comparison with Kalshi.  Paper~9 uses leg-specific protocol-finality and redemption clocks to model asynchronous event spreads.  This paper does not duplicate their mechanism-design or cross-venue analyses; it supplies the empirical protocol layer on which they depend.

\subsection{Roadmap}

\Cref{sec:related_work} positions the paper.  \Cref{sec:protocol_architecture} reconstructs the downstream architecture.  \Cref{sec:formal_framework} develops the condition, payout, and redemption model.  \Cref{sec:data_methodology} fixes the population, target-event registry, exact linkage, and validity gates.  \Cref{sec:estimands_design} defines protocol-finality, payout, survival, and observed-realization estimands.  \Cref{sec:economics_gap} restores the economics of the protocol-to-cash gap.  \Cref{sec:empirical_contract} gives the registered empirical design and analysis order, and \Cref{sec:results} presents the resulting numerical evidence.  The remaining sections state leveraged-design implications, reproducibility, limitations, and conclusions.  The appendices contain the operational dictionary, mapping schema, gate matrix, survival protocol, proofs, and reproducibility and evidence-lineage material.

\section{Related Work}
\label{sec:related_work}

The paper lies at the intersection of prediction-market design, oracle-mediated adjudication, conditional-token accounting, financial-market-infrastructure finality, and event-history analysis.

\subsection{Prediction markets and event claims}

Prediction-market research has traditionally focused on price interpretation, calibration, and information aggregation \citep{wolfers_zitzewitz_2004,manski_2006,hanson_2003}.  Those questions concern what prices mean before an event is resolved.  The present paper concerns what the protocol does after an upstream outcome has been adjudicated.  The distinction matters because a correct forecast price, a final oracle result, a recorded payout rule, and cash in a holder account are not the same object.

The companion programme develops the broader instrument context.  Paper~1 establishes why ordinary crypto-perpetual mechanics do not transfer cleanly to bounded event claims; Paper~2 provides the variant taxonomy; Paper~3 separates price, outcome, resolution, and informed-trading risks; Paper~4 states the observability limits of public fill-side data; and Part~I of the present study reconstructs Polymarket's semantic and oracle layer \citep{nechepurenko2026_p1,nechepurenko2026_p2,nechepurenko2026_p3,nechepurenko2026_p4,nechepurenko2026_p5a}.  Part~II takes the finality vocabulary downstream to protocol accounting and holder action.

\subsection{Oracles and stage separation}

Oracle research studies how off-chain facts enter on-chain systems, how proposals and disputes are organized, and how governance affects the integrity and latency of adjudication \citep{eskandari_oracles_2021,muehlberger_oracle_patterns_2020,adler_astraea_2018,kota_ai_oracles_2026}.  These works motivate the separation between evidence, proposal, challenge, and final answer.  They do not by themselves identify when a downstream token contract records a payout vector or when a holder realizes collateral.

Part~I shows empirically that technical request creation, first proposal, dispute, reset, oracle settlement, and adapter consumption are distinct clocks.  Part~II preserves that separation and adds protocol payout recording, redeemability, and observed redemption.  The object is therefore not a single ``resolution timestamp'' but a path through multiple institutional and contract layers.

\subsection{Conditional tokens and protocol accounting}

Conditional tokens provide a composable mechanism for preparing conditions, constructing outcome positions, recording payout numerators, and redeeming positions against collateral \citep{gnosis_conditional_tokens_2020,polymarket_ctf_overview_2026,polymarket_positions_tokens_2026}.  The complete-set construction implies a conservation structure: complementary positions are minted against collateral, can be merged before finality, and are redeemed according to the recorded payout vector afterward.  The exact outstanding entitlement, however, depends on the full position and collateral history, not merely on lifetime trade notional or the count of redemption events.

This paper uses the contract event ledger to identify protocol states and observed holder actions.  It deliberately does not assume that the event ledger alone reconstructs all outstanding balances.  The value-denominator problem is treated as a separate gate rather than hidden inside a normalization.

\subsection{Financial-market-infrastructure finality}

The distinction between trade agreement, settlement finality, asset availability, and end-user use resembles the discipline applied in conventional financial-market infrastructures \citep{cpss_iosco_2012}.  The analogy is limited: CTF conditions are smart-contract objects rather than conventional securities-settlement instructions.  Nevertheless, the core accounting principle transfers.  An upstream decision is not cash, and a technically available asset is not necessarily realized value for every participant.

The literature on market and funding liquidity also motivates the economics of delayed realization \citep{brunnermeier_pedersen_2009,duffie_garleanu_pedersen_2005}.  A redeemable claim can trade or be valued below its eventual payout because of capital lock-up, execution cost, uncertainty, or segmentation.  This paper does not impose one equilibrium model of that wedge.  It identifies the states and durations required for such models to be tested.

\subsection{Event history, censoring, and multistate analysis}

First redemption is a time-to-event outcome observed only after protocol finality.  Conditions without a qualifying redemption by the frozen snapshot are right-censored rather than assigned an infinite delay.  Kaplan--Meier, interval-censoring, and multistate methods provide the relevant statistical language \citep{kaplan_meier_1958,turnbull_1976,aalen_johansen_1978,andersen_borgan_gill_keiding_1993,fine_gray_1999}.  The application here requires two additional precautions.  First, a zero-payout redemption and a positive-payout redemption are ordered outcomes rather than mutually exclusive terminal causes.  Second, same-transaction events can have zero wall-clock separation while retaining exact log-order differences.

\subsection{Boundary with Papers~6, 8, and 9}

Paper~6 asks how an event-linked derivative should react to finality states and whether a synthetic position can be converted into funded conditional-token exposure.  Paper~8 compares the functional finality architecture of Polymarket and Kalshi.  Paper~9 studies asynchronous finality in multi-leg event spreads.  The present paper does not solve those design problems.  It defines and measures the downstream protocol states they require.

\section{Protocol Architecture: Conditions, Payouts, and Redemption}
\label{sec:protocol_architecture}

This section reconstructs the downstream mechanism from the pinned Conditional Tokens source, adapter source, and current public documentation.  It separates documentary description, source-code capability, deployed configuration, and observed historical use.  No claim is promoted from one layer to another without evidence.

\subsection{Evidence layers}

The protocol analysis uses four evidence layers.

\begin{enumerate}
  \item \textbf{Documentary description:} what the venue and protocol documentation state about preparation, resolution, position tokens, and redemption.
  \item \textbf{Source-code capability:} what the pinned contract and adapter revisions permit, including branches that may never have been exercised.
  \item \textbf{Deployed configuration:} the contract addresses, bytecode, dependencies, and ABI that applied to the frozen historical interval.
  \item \textbf{Observed use:} the actual events, calls, and canonical chain identities in the fixed snapshot.
\end{enumerate}

A current interface state is not projected backward into history.  A source branch establishes capability but not use.  A stream label does not override the actual event topic.  These distinctions became especially important during the construction of the protocol dataset, because historical transport labels and final materialized tables did not always preserve the complete semantic envelope.  The publication layer therefore derives event identity from the pinned ABI and full topics/data rather than trusting legacy names.

\begin{table}[htbp]
\centering
\caption{Evidence layers used in the protocol-finality analysis.  A stronger claim requires evidence from the corresponding layer rather than an inference from a weaker layer.}
\label{tab:evidence_layers}
\small
\begin{tabularx}{\textwidth}{L{0.22\textwidth}L{0.29\textwidth}Y}
\toprule
Layer & Primary object & Claim boundary\\
\midrule
Documentary & current public documentation & mechanism description, not historical use\\
Source capability & pinned contract and adapter source & permitted transition, not exercised transition\\
Deployed configuration & address, bytecode, ABI, dependency graph & historically applicable contract surface\\
Observed protocol & canonical chain events and full log envelopes & event occurrence, order, and decoded fields\\
Holder realization & canonical redemption events and balances where available & observed action; not beneficial ownership or complete entitlement\\
\bottomrule
\end{tabularx}
\end{table}

\subsection{Condition preparation}

A CTF condition is associated with an oracle, a question identifier, and an outcome-slot count.  Preparation registers the relation from these inputs to the condition identifier.  It is a protocol registry event, not a statement that the external event has occurred, that the rule is decidable, or that a payout is known.

For the frozen adapter cohort, preparation is useful for three reasons.  It supplies the exact question-to-condition relation; it fixes the outcome-slot structure; and it provides a protocol-side clock that can be compared with initialization and later finality.  Preparation can precede the first oracle request by a long interval and therefore does not enter the redemption risk set.

\subsection{Oracle availability, adapter consumption, and payout recording}

When a valid oracle value is available, a caller can invoke the adapter's downstream resolution path.  In the standard source path, the adapter obtains or settles the oracle value, constructs payout numerators, calls the CTF payout-reporting function, and emits its own terminal event \citep{polymarket_uma_adapter_contract_2026,polymarket_uma_adapter_repo_2026}.  The logical layers remain distinct even when the events share one transaction:

\begin{align*}
  \toracle &: \text{the successful oracle value is irreversible and consumable},\\
  \tconsume &: \text{the adapter consumes the value},\\
  \tresolve &: \text{the protocol records payout numerators},\\
  \tadapter &: \text{the adapter emits its observed terminal event}.
\end{align*}

No universal timestamp ordering is imposed between \(\tresolve\) and \(\tadapter\).  If both are emitted in one transaction, canonical log indices determine the observed order.  If they occur in different transactions or branches, block and transaction identity determine the order.  The relevant empirical object is the exact downstream path, not an assumed sequence copied from a schematic diagram.

\subsection{Core protocol events}

The publication dataset is built from event families derived from the pinned ABI rather than from mutable labels.  The three target event families and their analytical roles are summarized in \Cref{tab:protocol_events}.

\begin{table}[htbp]
\centering
\caption{Core protocol event families and their analytical role.  Exact field names and indexed layout are frozen from the pinned ABI in the release registry.}
\label{tab:protocol_events}
\small
\begin{tabularx}{\textwidth}{L{0.28\textwidth}L{0.30\textwidth}Y}
\toprule
Event & Load-bearing fields & Analytical role\\
\midrule
\makecell[l]{Condition preparation\\{\scriptsize\texttt{ConditionPreparation}}} & oracle, question identifier, outcome-slot count, condition identifier relation & Exact registry and question-to-condition mapping\\
\makecell[l]{Condition resolution\\{\scriptsize\texttt{ConditionResolution}}} & oracle, question identifier, condition identifier, outcome-slot count, ordered payout numerators & Protocol-finality time and payout-vector taxonomy\\
\makecell[l]{Payout redemption\\{\scriptsize\texttt{PayoutRedemption}}} & redeemer, collateral token, parent collection, condition identifier, index sets, payout & Holder-action timing and observed collateral realization\\
\bottomrule
\end{tabularx}
\end{table}

Preparation creates the registry relation.  Condition resolution records the ordered payout numerators.  Payout redemption records one observed holder action.  Source observations from providers remain provenance records; repeated source observations do not create additional chain events.

\subsection{Payout vectors and protocol finality}

Let \(K_c\) be the number of outcome slots for condition \(c\), and let
\[
  \Yvec_c=(Y_{c1},\ldots,Y_{cK_c})
\]
be the ordered payout numerators.  The sum
\[
  D_c=\sum_{k=1}^{K_c}Y_{ck}
\]
forms the payout denominator when \(D_c>0\).  The normalized payout vector is
\[
  \pivec_c=\Yvec_c/D_c.
\]

Protocol finality is defined by the first valid event that makes this ordered rule effective for redemption.  Raw integer numerators remain in the canonical data because scale, ordering, and malformed representations are protocol facts.  Economic taxonomy is applied only after lossless decoding and validity checks.

For a binary condition, canonical classes include YES, NO, and a neutral 50/50 outcome.  Other positive numerator patterns can be valid under a broader slot structure or non-standard route.  A zero denominator, inconsistent slot count, malformed array, or conflicting resolution is not silently normalized into a familiar class.

\subsection{Collateral conservation and the value denominator}

Outcome positions are created against collateral through complete-set splitting, can be merged before finality, and are redeemed against the recorded payout rule afterward \citep{polymarket_positions_tokens_2026,polymarket_ctf_overview_2026}.  This supplies a conservation structure but not an automatic value denominator for the empirical paper.

The amount economically outstanding at protocol finality is not lifetime trading volume, token market capitalization, or the sum of all transfer events.  It depends on net splits, merges, wrapper transformations, collateral routes, fees, burns, and balances.  A complete entitlement reconstruction must therefore follow balance-consistent position and collateral flows.  Until that reconstruction passes its own gate, the paper reports absolute observed payout amounts and event-based timing, not a fraction of all entitlement redeemed.

\subsection{Redeemability and observed cash realization}

After payout reporting, holders can redeem positions for collateral through the applicable CTF and collateral path \citep{polymarket_resolution_docs_2026,polymarket_positions_tokens_2026,polymarket_ctf_overview_2026}.  Protocol redeemability and observed redemption are different outcomes:
\[
  \tredeem(c)\geq\tresolve(c)
\]
for every accepted redemption event under the observed transition rules.

A long observed gap can reflect user inattention, batching, transaction cost, custody operations, account abstraction, gas sponsorship, or a technical constraint.  The paper therefore reports both protocol finality and holder-action timing.  It does not label the entire holder delay as an oracle or protocol delay.

The first observed redemption is a condition-level population statistic, not a universal transition for every holder.  Holder-specific cash realization is an address/position outcome when the necessary balances and redemption logs are available.  Absence of a redemption event does not mean protocol finality failed.

\subsection{Protocol state representation}

\begin{figure}[htbp]
\centering
\resizebox{0.97\textwidth}{!}{%
\begin{tikzpicture}[
  node distance=8mm and 10mm,
  every node/.style={font=\footnotesize},
  st/.style={seriesbox,minimum width=2.45cm,minimum height=10mm},
  terminal/.style={seriesbox,minimum width=2.45cm,minimum height=10mm,fill=seriesblue!10},
  cens/.style={seriesbox,dashed,minimum width=2.45cm,minimum height=10mm,fill=white},
  arr/.style={seriesarrow}
]
\node[st] (prepared) {Prepared\\condition};
\node[st,right=of prepared] (oracle) {Oracle-final\\request};
\node[terminal,above right=6mm and 14mm of oracle] (protocol) {Payout vector\\recorded};
\node[st,below right=6mm and 14mm of oracle] (adapter) {Adapter-terminal\\log};
\node[st,right=16mm of protocol] (redeemable) {Redeemable\\positions};
\node[terminal,above right=7mm and 12mm of redeemable] (positive) {First positive\\redemption};
\node[terminal,below right=7mm and 12mm of redeemable] (zero) {First zero-payout\\redemption};
\node[cens,below=18mm of redeemable] (censored) {No observed redemption\\by fixed snapshot};
\draw[arr] (prepared) -- (oracle);
\draw[arr] (oracle) -- (protocol);
\draw[arr] (oracle) -- (adapter);
\draw[<->,thick,dashed,seriesblue] (protocol) -- node[right,font=\scriptsize,align=left] {relative emission order\\measured, not assumed} (adapter);
\draw[arr] (protocol) -- (redeemable);
\draw[arr] (redeemable) -- (positive);
\draw[arr] (redeemable) -- (zero);
\draw[arr,dashed] (redeemable) -- (censored);
\node[draw=black!50,dashed,rounded corners=2pt,fit=(protocol)(adapter),inner sep=4mm,label={[font=\scriptsize]below:downstream mechanism transaction or sequence}] {};
\end{tikzpicture}%
}
\caption{Protocol and redemption state representation.  Oracle finality is followed by downstream mechanism actions that can emit both the protocol-resolution and adapter-terminal logs; no universal order between those emitted logs is imposed.  The payout vector creates redeemability.  Positive and zero-payout redemption are distinct observed outcomes, while absence of an event at the snapshot is right censoring rather than a terminal failure state.}
\label{fig:protocol_state_machine}
\end{figure}
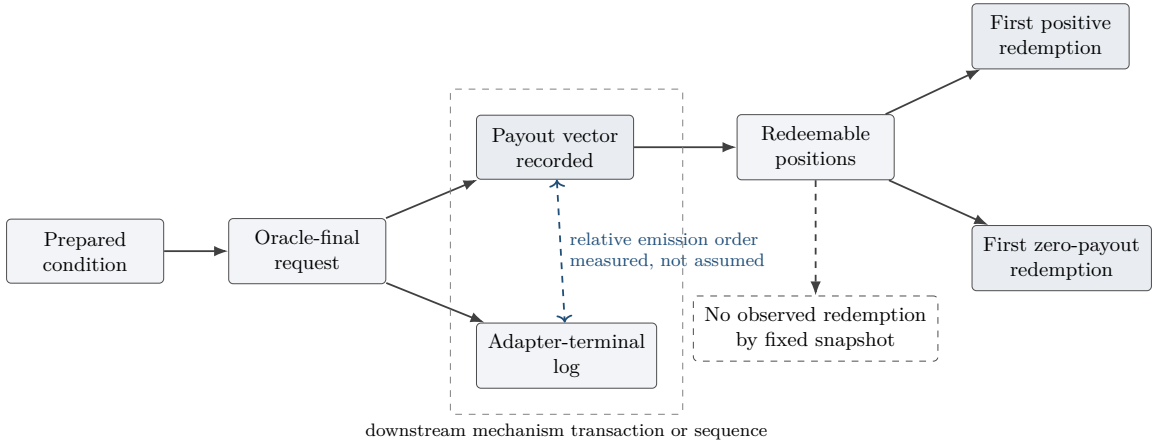

The state representation distinguishes mechanism-controlled and holder-controlled transitions.  Oracle finality and payout recording are mechanism states.  Redemption is an action under an already recorded rule.  A condition can therefore be protocol-final and right-censored for observed redemption without contradiction.

\subsection{Architecture summary}

The downstream stack contains at least five distinct economic layers:

\begin{enumerate}
  \item a prepared condition and outcome-slot structure;
  \item an irreversible oracle result;
  \item adapter consumption and adapter-terminal emission;
  \item a recorded protocol payout rule and technical redeemability;
  \item observed holder conversion of positions into collateral.
\end{enumerate}

The rest of the paper formalizes these layers, reconstructs them from exact events, and defines which claims each layer can support.

\section{Formal Framework: Conditions, Payouts, and Holder Action}
\label{sec:formal_framework}

Let \(q\in\Cohort\) index a frozen adapter-question, \(r\) an exact oracle request generation, \(c\in\Conditions\) a CTF condition, and \(j\in\Redemptions_c\) a redemption event linked to condition \(c\).  Let \(\phi(q)=c\) denote the exact question-to-condition mapping when identified.  The mapping remains partial until the linkage gate passes.

For condition \(c\), let \(K_c\geq2\) be the outcome-slot count and let
\[
  \Yvec_c=(Y_{c1},\ldots,Y_{cK_c}),
  \qquad
  D_c=\sum_{k=1}^{K_c}Y_{ck}.
\]
When \(D_c>0\), the normalized payout vector is
\[
  \pivec_c=\frac{\Yvec_c}{D_c}.
\]
All numerators remain exact integers in the canonical data.  Floating-point conversion is not used for identity, classification, or accounting.

\begin{assumption}[Pinned transition semantics]
\label{ass:pinned-transition-semantics}
The accepted contract and adapter artifacts correctly represent the transition preconditions used in the empirical interval.  In particular, an accepted redemption requires a previously recorded payout rule for the same condition.
\end{assumption}

\begin{assumption}[Exact load-bearing event identity]
\label{ass:exact-load-bearing-event-identity}
Analytical events have a valid chain identifier, contract address, transaction hash, exact log index, complete ordered topics, and complete data.  A nullable transaction index is auxiliary and does not replace or weaken the exact log index.
\end{assumption}

\begin{assumption}[Fixed administrative censoring]
\label{ass:fixed-administrative-censoring}
The follow-up boundary \(\tstar\) is fixed before holder-event analysis.  Events after the boundary are unobserved for this release and are not interpreted as absent forever.
\end{assumption}

\subsection{Protocol and holder clocks}

\begin{definition}[Preparation time]
\label{def:preparation-time}
\(\tprepare(c)\) is the timestamp of the canonical condition-preparation event linked to \(c\).  It is a registry time and need not be close to event occurrence, oracle request, or finality.
\end{definition}

\begin{definition}[Protocol-finality time]
\label{def:protocol-finality-time}
\(\tresolve(c)\) is the timestamp of the first valid condition-resolution event whose ordered payout numerators become the protocol's recorded redemption rule for \(c\).
\end{definition}

\begin{definition}[Adapter-terminal emission time]
\label{def:adapter-terminal-emission-time}
\(\tadapter(c)\) is the timestamp of the canonical terminal event emitted by the requesting adapter and linked to \(c\) through the exact Part~I relation.  It is an observed log time, not a proxy for the internal beginning of the call.
\end{definition}

\begin{definition}[First observed redemption]
\label{def:first-observed-redemption}
\(\tredeem(c)\) is the earliest timestamp among canonical redemption events linked to \(c\).  If no qualifying event is observed by the fixed follow-up boundary \(\tstar\), the endpoint is right-censored.
\end{definition}

\begin{definition}[First observed positive-payout redemption]
\label{def:first-observed-positive-payout-redemption}
\(\tredeempos(c)\) is the earliest canonical redemption event for \(c\) whose paid collateral is strictly positive.  A zero-payout redemption can therefore precede \(\tredeempos(c)\).
\end{definition}

The ordinary load-bearing partial order is
\begin{equation}
  \toracle(c)\leq\tresolve(c)\leq t_{cj}
  \label{eq:ordinary_partial_order}
\end{equation}
for each accepted redemption event \(j\), and
\[
  \toracle(c)\leq\tadapter(c)
\]
when the exact upstream endpoint is observed.  No universal ordering is imposed between \(\tresolve(c)\) and \(\tadapter(c)\).

\subsection{Exact ordering beyond wall-clock time}

Each canonical event has an order key
\[
  \kappa(e)=
  \bigl(
  b_e,\, i_e^{\mathrm{tx}},\, i_e^{\mathrm{log}}
  \bigr),
\]
where \(b_e\) is block number, \(i_e^{\mathrm{tx}}\) is nullable transaction index, and \(i_e^{\mathrm{log}}\) is exact log index.  The canonical event order uses
\[
  (b_e,i_e^{\mathrm{log}},h_e,a_e),
\]
with transaction hash \(h_e\) and contract address \(a_e\) as deterministic tie-break fields.  A missing transaction index is retained as null and does not block the event.  A missing log index is load-bearing and must be recovered exactly or excluded from canonical ordering.

\begin{proposition}[Same-transaction zero duration does not erase state order]
\label{prop:same_tx_order}
Suppose adapter terminal transition and condition resolution occur in one transaction.  Their wall-clock duration is zero at block-timestamp precision, but their state order remains identified when their canonical log indices differ.
\end{proposition}

\begin{proof}
Both events inherit the same block timestamp and transaction hash.  Canonical log identity includes the exact log index, which induces a strict order within the transaction.  Thus the duration measured in seconds is zero while the transition order is non-degenerate.
\end{proof}

\subsection{Condition-level state}

Define the condition-level state vector
\[
  X_c(t)=\bigl(P_c(t),R_c(t),H_c(t)\bigr),
\]
where \(P_c(t)\) records whether the condition is prepared, \(R_c(t)\) records whether a valid payout vector has been recorded, and \(H_c(t)\) is the history of observed redemption events through \(t\).  The state is event-sourced: it is reconstructed from the ordered canonical history rather than read from a mutable terminal snapshot.

A minimal state path is
\[
  \text{Unprepared}
  \rightarrow
  \text{Prepared}
  \rightarrow
  \text{ProtocolFinal}
  \rightarrow
  \text{Redeemable},
\]
with holder-action branches
\[
  \text{NoRedemption}
  \rightarrow
  \text{ZeroOnlyObserved}
  \rightarrow
  \text{PositiveObserved},
\]
and a possible direct transition from \(\text{NoRedemption}\) to \(\text{PositiveObserved}\).  Administrative censoring at \(\tstar\) is an observation boundary, not an absorbing protocol state.

\subsection{Separation of finality layers}

\begin{proposition}[Oracle finality does not identify protocol finality]
\label{prop:oracle_protocol_nonequiv}
There exist two histories with the same oracle-final timestamp and terminal oracle value but different protocol-finality times.
\end{proposition}

\begin{proof}
In both histories the same oracle request settles at \(\toracle\) with the same value.  In the first history the adapter consumes the result and records the payout vector in the same transaction.  In the second, protocol consumption occurs in a later transaction.  The upstream oracle observation is identical, while \(\tresolve\) differs.  Hence oracle finality does not identify protocol finality.
\end{proof}

\begin{proposition}[Protocol finality does not identify holder realization]
\label{prop:protocol_holder_nonequiv}
There exist two histories with the same condition, payout vector, and protocol-finality time but different redemption histories.
\end{proposition}

\begin{proof}
Fix one valid condition-resolution event.  In the first history a holder redeems immediately after protocol finality.  In the second, no holder redeems by \(\tstar\).  The protocol state and payout vector are identical, while \(H_c(t)\) and \(\tredeem(c)\) differ.  Therefore protocol finality does not identify holder realization.
\end{proof}

\begin{corollary}[A scalar resolution timestamp is insufficient]
\label{cor:scalar_timestamp_insufficient}
Any scalar timestamp that identifies only oracle finality, adapter terminality, or protocol finality fails to identify the full holder-action path whenever redemption timing enters the economic outcome.
\end{corollary}
\begin{proof}
A scalar endpoint fixes at most one finality coordinate. By Propositions~\ref{prop:oracle_protocol_nonequiv} and~\ref{prop:protocol_holder_nonequiv}, histories can agree on that scalar coordinate while differing in downstream protocol or holder-action state. Hence the scalar timestamp cannot identify the full path whenever redemption timing matters.
\end{proof}

\subsection{Payout representation}

\begin{proposition}[Payout-vector scale invariance]
\label{prop:payout_scale}
Assume \(D_c>0\). For any scalar \(a>0\), payout numerators \(\Yvec_c\) and \(a\Yvec_c\) define the same normalized payout vector whenever both are valid protocol representations.
\end{proposition}

\begin{proof}
The normalized vector under the scaled representation is
\[
\frac{a\Yvec_c}{\sum_k aY_{ck}}
=
\frac{a\Yvec_c}{aD_c}
=
\frac{\Yvec_c}{D_c}
=
\pivec_c.
\]
The raw numerator representation remains relevant for protocol validation, but the economic per-unit payout class is unchanged.
\end{proof}

For binary conditions the normalized classes include
\[
  (1,0),\qquad (0,1),\qquad \left(\tfrac12,\tfrac12\right),
\]
plus any other valid normalized pattern admitted by the frozen slot and route structure.  Classification occurs only after verifying slot count, positive denominator, ordered numerators, and exact event identity.

\subsection{Redemption and entitlement}

Let \(a_{cj}\geq0\) be the collateral paid in redemption event \(j\), and define observed cumulative payout after protocol finality by
\begin{equation}
  A_c(u)=
  \sum_{j\in\Redemptions_c}
  a_{cj}\,
  \ind{0\leq t_{cj}-\tresolve(c)\leq u}.
  \label{eq:observed_payout_curve}
\end{equation}
This is an identified event-ledger quantity when event coverage and decoding pass.

Let \(W_c\) denote total winning-position entitlement at protocol finality, measured in the same collateral unit.  The entitlement-completion ratio would be
\begin{equation}
  \Gamma_c(u)=\frac{A_c(u)}{W_c}.
  \label{eq:entitlement_completion}
\end{equation}

\begin{theorem}[Redemption-event insufficiency for entitlement completion]
\label{thm:entitlement_nonidentification}
A complete canonical ledger of \emph{redemption events alone} does not identify \(\Gamma_c(u)\) unless \(W_c\) is independently identified. This does not rule out identification from a broader balance-complete reconstruction of splits, merges, transfers, wrappers, and outstanding ERC-1155 positions.
\end{theorem}

\begin{proof}
Fix an observed redemption history and hence \(A_c(u)\).  Two admissible balance histories can generate the same redemption events but different outstanding winning-position entitlement at protocol finality, for example because unredeemed winning positions differ.  Then the numerator is identical while \(W_c\) and therefore \(\Gamma_c(u)\) differ.  The event ledger alone cannot select between the histories.
\end{proof}

\begin{figure}[htbp]
\centering
\resizebox{0.96\textwidth}{!}{%
\begin{tikzpicture}[
  node distance=8mm and 12mm,
  every node/.style={font=\footnotesize},
  obs/.style={seriesbox,minimum width=3.0cm,minimum height=11mm},
  missing/.style={seriesbox,dashed,fill=white,minimum width=3.2cm,minimum height=11mm},
  arr/.style={seriesarrow}
]
\node[obs] (events) {Canonical redemption\\events};
\node[obs,right=of events] (paid) {Observed collateral\\paid};
\node[obs,right=of paid] (curve) {Observed cumulative\\payout curve $A_c(u)$};
\node[missing,below=14mm of paid] (balances) {Complete winning-position\\balances / entitlement $W_c$};
\node[missing,right=of balances] (ratio) {Entitlement completion\\$A_c(u)/W_c$};
\draw[arr] (events)--(paid);
\draw[arr] (paid)--(curve);
\draw[arr,dashed] (balances)--(ratio);
\draw[arr,dashed] (curve)--(ratio);
\node[font=\scriptsize,seriesblue,align=center] at ($(paid.south)+(0,-8mm)$) {Observed and identified};
\node[font=\scriptsize,align=center] at ($(ratio.south)+(0,-8mm)$) {Not identified without a separate value gate};
\end{tikzpicture}%
}
\caption{Observation-versus-entitlement boundary.  Redemption events identify realized collateral and its timing.  The fraction of all entitlement redeemed additionally requires complete winning-position balances, which are not implied by the event ledger.}
\label{fig:observation_entitlement_boundary}
\end{figure}
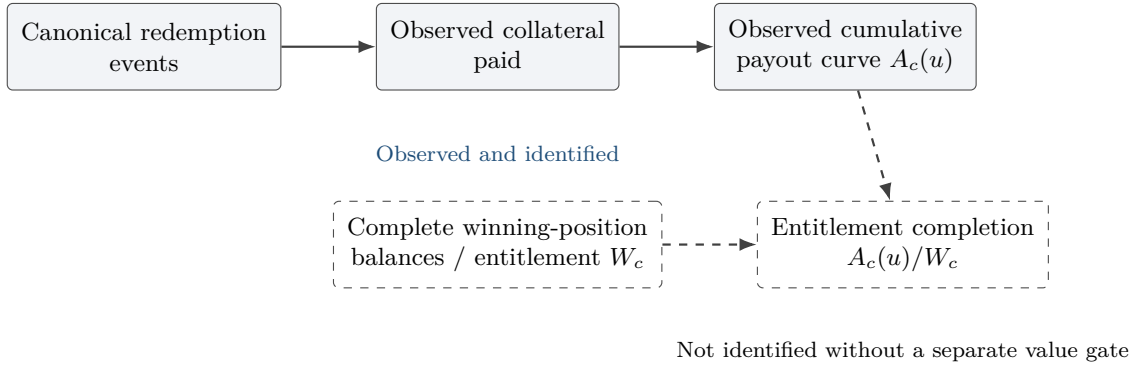

\subsection{Source observations and canonical events}

Let \(o\in\sourceobs\) denote a provider source observation and \(e\in\canon\) a canonical on-chain event.  Several source observations can refer to the same event because of overlapping ranges, repeated provider delivery, recovery, or transport duplication.

\begin{proposition}[Source duplication does not increase event multiplicity]
\label{prop:source_duplication}
Suppose several source observations share one exact canonical identity and one semantic payload, and receipt or chain evidence contains one corresponding log.  Then the correct representation is one canonical event with several provenance relations.
\end{proposition}

\begin{proof}
The chain identity and semantic payload establish one on-chain event.  Repeated source delivery changes transport multiplicity but not the chain history.  One event plus all provenance relations therefore preserves both facts without inflating event counts.
\end{proof}

\subsection{Clock decomposition}

For exactly linked conditions define
\begin{align}
  L_c^{OP} &= \tresolve(c)-\toracle(c), \\
  L_c^{OA,C} &= \tadapter(c)-\toracle(c), \\
  D_c^{PA} &= \tadapter(c)-\tresolve(c), \\
  L_c^{PR} &= \tredeem(c)-\tresolve(c), \\
  L_c^{PR+} &= \tredeempos(c)-\tresolve(c).
\end{align}
The signed adapter--protocol gap \(D_c^{PA}\) is accompanied by a categorical order class when endpoints share a timestamp.  The redemption intervals are censored when the endpoint is not observed by \(\tstar\).

\begin{figure}[htbp]
\centering
\resizebox{0.92\textwidth}{!}{%
\begin{tikzpicture}[
  node distance=14mm and 22mm,
  every node/.style={font=\footnotesize},
  st/.style={seriesbox,minimum width=2.6cm,minimum height=10mm},
  arr/.style={-{Latex[length=2mm]},thick,draw=black!75}
]
\node[st] (oracle) {Oracle final};
\node[st,above right=8mm and 25mm of oracle] (protocol) {Protocol final};
\node[st,below right=8mm and 25mm of oracle] (adapter) {Adapter terminal};
\node[st,right=28mm of protocol] (redeem) {First observed\\redemption};
\draw[arr] (oracle) -- node[above left,font=\scriptsize] {$L^{OP}$} (protocol);
\draw[arr] (oracle) -- node[below left,font=\scriptsize] {$L^{OA}$} (adapter);
\draw[<->,thick,dashed,seriesblue] (protocol) -- node[right,font=\scriptsize,align=left] {$D^{PA}$ signed gap;\\log-order class if tied} (adapter);
\draw[arr] (protocol) -- node[above,font=\scriptsize] {$L^{PR}$} (redeem);
\end{tikzpicture}%
}
\caption{Measurement graph for downstream finality clocks.  Oracle-to-protocol and oracle-to-adapter intervals are estimated separately.  The adapter--protocol comparison is signed and accompanied by a transaction/log-order class; no universal direction is assumed.  Protocol-to-redemption is an observed holder-action interval and requires right-censoring methods.}
\label{fig:clock_decomposition}
\end{figure}
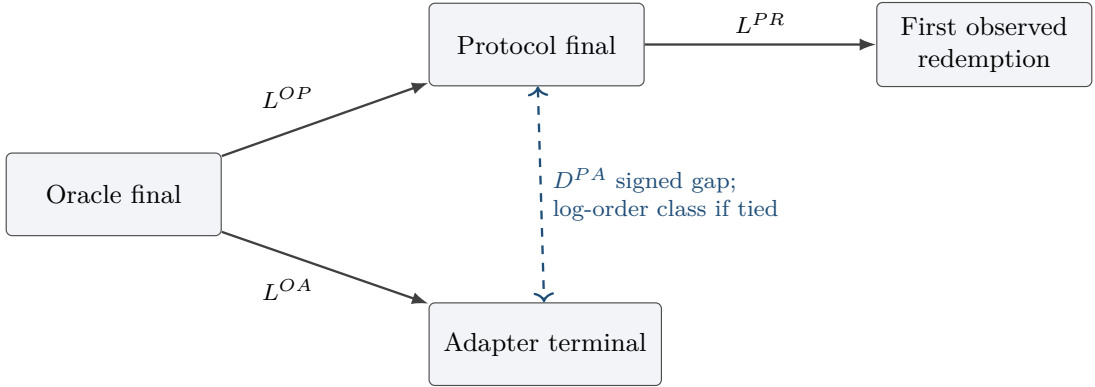

\subsection{State occupation and capital time}

For a protocol or holder state \(s\), define occupation over \([a,b]\) by
\begin{equation}
  \occupancy_{cs}(a,b)
  =
  \int_a^b \ind{X_c(t)=s}\,dt.
  \label{eq:condition_state_occupation}
\end{equation}
Two conditions can share the same first and last timestamps while spending different time in zero-only, positive-observed, or unresolved holder states.  Occupation time is therefore a separate input to capital-lock and service-cost analysis.

\subsection{Observation grades}

Each clock receives one coordinate-specific grade:
\begin{description}[style=nextline,leftmargin=2.6cm]
\item[$\ObsExact$ --- Exact] Exact canonical event time and deterministic order.
\item[$\ObsInterval$ --- Interval] Endpoint bounded by exact observations but not point-identified.
\item[$\ObsSnapshot$ --- Snapshot] Current or terminal representation observed without transition history.
\item[$\ObsProxy$ --- Proxy] Documented substitute with an explicit relationship to the target coordinate.
\item[$\ObsUnmeasured$ --- Unmeasured] Evidence does not identify the coordinate.
\item[$\ObsConflict$ --- Conflicting] Load-bearing evidence disagrees.
\end{description}

A convenient upstream or downstream timestamp is not substituted for a missing coordinate.  Duration tables include only compatible grades or report interval bounds explicitly.

\section{Data, Population Construction, and Identification}
\label{sec:data_methodology}

The empirical design inherits the frozen adapter-question cohort from Part~I and adds a fixed downstream follow-up boundary.  It does not redefine the population from a mutable application-programming-interface count or from the number of protocol events returned by one provider.

\subsection{Frozen population and follow-up}

Let
\[
  \Cohort
  =
  \{q:\text{adapter-question initialized by block }79{,}721{,}080\}
\]
be the frozen question cohort.  The accepted Part~I denominator contains \(185{,}550\) adapter-questions.  Downstream protocol and redemption events are followed through Polygon block \(90{,}114{,}204\).  Events after the snapshot do not enter the analysis, and questions initialized after the cutoff do not enter the cohort even if their protocol events occur before the follow-up boundary.

This asymmetric design separates cohort formation from follow-up.  It permits late protocol finality and redemption for already frozen questions without allowing the denominator to expand during analysis.

\subsection{Contract-wide acquisition universe and cohort restriction}

The exact target-event scan covers the pinned Conditional Tokens contract over the registered deployment-to-snapshot interval.  That contract-wide event universe serves three purposes: it proves event-family coverage, preserves canonical event identity, and supplies the pool from which exact cohort links are drawn.  It is not itself the study population.

Cohort-specific claims require an exact bridge
\[
\begin{aligned}
  \text{adapter-question}
  &\longrightarrow \text{question identifier}
  \longrightarrow \text{condition identifier}\\
  &\longrightarrow \text{protocol and redemption events}.
\end{aligned}
\]
An event that cannot be linked through this bridge can remain valid protocol evidence while remaining outside the Polymarket cohort analysis.  Conversely, a cohort question with no identified downstream condition remains in the attrition denominator rather than disappearing from the study.  This separation prevents full-contract counts from being misreported as question-level coverage or redemption rates.

\begin{table}[htbp]
\centering
\caption{Separation of the contract-wide acquisition universe from the frozen Polymarket study cohort.}
\label{tab:universe_boundary}
\small
\begin{tabularx}{\textwidth}{L{0.25\textwidth}L{0.29\textwidth}Y}
\toprule
Layer & Unit and authority & Permitted use\\
\midrule
Contract-wide target-event universe & Canonical CTF events selected by ABI-derived topic and exact chain identity over the fixed block interval & Coverage, transport integrity, event semantics, and candidate protocol relations\\
Frozen adapter-question cohort & Adapter-questions initialized by the Part~I cutoff & Scientific denominator for Polymarket mapping and downstream attrition\\
Exact cohort-linked condition set & Conditions linked through protocol identifiers and accepted lineage & Cohort-level protocol-finality, payout, and redemption estimands\\
Reconstruction-only remainder & Valid protocol events not accepted into the frozen cohort relation & Protocol-universe diagnostics; never substituted for cohort rates\\
\bottomrule
\end{tabularx}
\end{table}

\subsection{Nested observation universes}

The paper uses nested universes rather than one denominator for every result.

\begin{table}[htbp]
\centering
\caption{Nested observation universes.  Counts and exclusions are reported separately at every transition; no lower-stage numerator is silently divided by a higher-stage denominator.}
\label{tab:nested_populations}
\small
\begin{tabularx}{\textwidth}{L{0.27\textwidth}L{0.25\textwidth}Y}
\toprule
Universe & Unit & Claims supported\\
\midrule
Frozen question cohort & adapter-question & cohort composition and mapping attrition\\
Exact prepared conditions & condition & registry coverage and condition multiplicity\\
Protocol-final conditions & valid condition-resolution event & payout taxonomy and protocol-finality timing\\
Upstream-linked protocol-final conditions & condition with exact/graded Oracle and adapter links & finality-path decomposition\\
Redemption risk set & protocol-final condition with follow-up & first-redemption and censoring estimands\\
Positive-redemption risk set & protocol-final condition with ordered redemption history & first-positive and zero-to-positive transitions\\
Value-qualified conditions & condition with exact entitlement denominator & entitlement-completion claims only\\
\bottomrule
\end{tabularx}
\end{table}

A lower-stage count is never divided by a higher-stage denominator without displaying the attrition path.  For example, first-redemption survival uses protocol-final conditions with valid follow-up, not all initialized adapter-questions.

\subsection{Authoritative target-event registry}

The target event registry is derived from the pinned ABI and source, not copied from a historical stream label.  For each event it records:

\begin{itemize}
  \item canonical Solidity event signature;
  \item ABI-derived topic hash;
  \item indexed and non-indexed field layout;
  \item contract address and deployment boundary;
  \item ABI and source hashes;
  \item decoder revision.
\end{itemize}

The target streams are condition preparation, condition resolution, and payout redemption.  Any other Conditional Tokens event remains useful protocol evidence but does not enter a target-event count under an incorrect label.

\subsection{Exact coverage and range lineage}

Population-wide event claims require an exact zero-gap range ledger for each target event through the fixed snapshot.  A successful terminal range must record the exact query coordinates, provider, response status, count, parent/child split lineage, page semantics, response hash, and output artifact.  A broad query that returns the provider cap is non-terminal and cannot prove coverage.

The scan uses recursive block splitting for capped responses.  Zero-event ranges remain evidence only when the provider successfully returned a terminal response.  Gaps, overlaps, pending branches, provider changes, and single-block pagination are explicit tables rather than implicit assumptions.

A provider response is transport evidence, not the research event itself.  The analytical layer retains complete log envelopes and canonical identities in compact columnar form.  Ordinary temporary responses can be discarded after verified commit, while a deterministic audit sample and all anomalies remain retained.

\subsection{Diagnostic supersession and evidentiary revision}

Historical diagnostics are not erased when later evidence resolves them.  The evidence record retains the original diagnostic, its cause, the corrective or interpretive disposition, and the later artifact supporting the revision.  A historical failure can therefore remain visible without governing a conclusion once its predicate has been re-evaluated under an explicit normalized schema or stronger evidence record.  This rule prevents both silent deletion of adverse evidence and accidental use of an obsolete diagnostic.

\subsection{Source observations and canonical events}

One source observation is one provider/page occurrence of a log.  One canonical event is one exact on-chain identity:
\[
  (\text{chain},\text{contract},\text{transaction hash},\text{log index}).
\]

The full envelope includes block number and hash, block timestamp, nullable transaction index with a raw-status field, exact log index, exact \code{removed} value, complete ordered topics, complete data, event class, lossless decoded fields, and provenance.  Repeated source observations map many-to-one to canonical events where semantic payloads agree.

A missing transaction index is non-load-bearing and remains null.  A missing log index blocks the affected event until exact receipt or block evidence recovers it.  No event is ordered from provider page position or source row order.

\subsection{Exact question-to-condition mapping}

The canonical mapping relation uses protocol identities only.  For each adapter-question it attempts to identify:

\begin{itemize}
  \item exact question identifier;
  \item condition identifier;
  \item oracle address;
  \item outcome-slot count;
  \item canonical condition-preparation event;
  \item adapter and request lineage;
  \item mapping method and evidence hashes.
\end{itemize}

Title, slug, text similarity, and nearest-time matching are excluded from the canonical relation.  Unmatched and ambiguous questions remain explicit.  A question can have one condition, several historically distinct condition objects, or no identified condition under the frozen evidence; these are different attrition states.

\subsection{Protocol-finality linkage}

For conditions with a valid resolution event, the protocol relation stores ordered payout numerators, normalized class, canonical finality identity, and exact links to upstream oracle and adapter states where available.  Same-transaction paths preserve log order.  Same-block but different-transaction paths preserve transaction identities.  Later-block paths retain the full interval.

The upstream oracle time can be exact, interval-observed, or unavailable depending on the route.  An unavailable oracle endpoint does not prevent measurement of protocol-to-redemption time, but it prevents oracle-to-protocol claims for that condition.

\subsection{Redemption linkage}

Each canonical redemption event is linked to condition identifier, redeemer, collateral token, parent collection, index sets, payout, and protocol-finality state.  The paper distinguishes:

\begin{itemize}
  \item any redemption event;
  \item strictly positive-payout redemption;
  \item zero-payout redemption;
  \item multiple redemption events for one condition;
  \item ordered zero-to-positive histories;
  \item right-censored conditions with no observed redemption by the snapshot.
\end{itemize}

Addresses are pseudonymous execution identifiers, not beneficial owners.  Address-level analyses are therefore stated as address behavior unless an independent identity source exists.

\subsection{Lifecycle cross-classification}

Every condition in the contract-wide analytical table is assigned to one cell of the complete prepared/resolved/redeemed cross-classification.  The resulting eight-cell table is reported separately for the full protocol universe and the frozen cohort.  Each non-canonical path receives an explicit disposition, such as left truncation, pre-deployment state, repeated event, unresolved linkage, quarantine, snapshot censoring, or genuine alternative lifecycle.  Marginal totals alone are insufficient because they can hide resolved-without-observed-preparation or redeemed-without-observed-resolution states whose interpretation depends on the fixed observation boundary and exact event history.

\subsection{Bitemporal provenance}

Each source record stores valid-time fields from the chain and retrieval-time fields from the acquisition process.  Provider response time is never substituted for block time.  Historical stream labels and current user-interface metadata remain provenance attributes rather than event truth.

The data lineage preserves:

\begin{enumerate}
  \item immutable source observation;
  \item canonical event identity;
  \item ABI-derived event class;
  \item exact linkage relation;
  \item analytical table row;
  \item reported table or figure source hash.
\end{enumerate}

\subsection{Validity gates}

\begin{table}[htbp]
\centering
\caption{Cumulative validity gates for Part II.}
\label{tab:validity_gates}
\begin{tabularx}{\textwidth}{L{0.18\textwidth}L{0.32\textwidth}Y}
\toprule
Gate & Requirement & Claims unlocked\\
\midrule
G-ABI & Topic hashes and decoders derived from the pinned ABI & Correct event-family classification\\
G-COVERAGE & Zero-gap/no-unexplained-overlap range ledger through the snapshot & Population-wide event counts for the target stream\\
G-IDENTITY & Exact transaction hash and log index; complete envelope and provenance & Canonical event multiplicity and ordering\\
G-DECODE & Lossless typed ABI decode with raw topics/data retained & Payout-vector and event-field analysis\\
G-LINK & Exact question-to-condition and upstream-state linkage & Question-level finality paths and durations\\
G-TIMING & Valid state order and fixed administrative censoring boundary & Duration, survival, and competing-risk estimates\\
G-VALUE & Exact outstanding entitlement or balance denominator & Value-weighted entitlement-completion claims\\
\bottomrule
\end{tabularx}
\end{table}

The gates are cumulative.  A downstream result cannot repair an upstream failure.  For example, a plausible survival curve cannot establish that the event stream had zero-gap coverage; a complete event stream cannot establish entitlement completion without a value denominator.

\subsection{Missingness and censoring}

Missingness is classified by mechanism rather than represented by one null code.  Categories include:

\begin{itemize}
  \item not applicable because the upstream state never occurred;
  \item right-censored at the fixed snapshot;
  \item exact target event unavailable because coverage failed;
  \item linkage unresolved despite valid events;
  \item interval-observed endpoint;
  \item value denominator unavailable;
  \item conflicting semantic or identity evidence.
\end{itemize}

Administrative censoring is deterministic given the snapshot.  It does not imply that a later redemption never occurs.

\subsection{Value-denominator gate}

The event ledger identifies paid collateral in observed redemption actions.  A value-weighted completion claim additionally requires an exact condition-level entitlement denominator.  Candidate reconstructions must account for splits, merges, burns, wrappers, negative-risk routes, collateral tokens, and historical balances.  Until that relation passes \(\valuegate\), the paper reports absolute observed payouts and explicitly labelled normalizations only.

\subsection{Ethics and disclosure}

The protocol data are public chain records, but the analysis follows minimum-necessary disclosure.  Ranked redemption-address lists are not reported by default.  Raw addresses can remain in reconstruction-only or controlled research layers while reported tables aggregate by condition, route, payout class, or timing stratum.  No beneficial ownership, intent, or legal status is inferred from an address alone.

\section{Estimands, Censoring, and Empirical Design}
\label{sec:estimands_design}

The empirical objects are specified before numerical insertion.  Each estimand has a fixed unit, denominator or risk set, endpoint definition, observation grade, and gate consequence.

\subsection{Population and attrition}

The population flow begins with the frozen question cohort and proceeds through exact condition mapping, protocol finality, and observed redemption.  Let
\[
  N_Q=|\Cohort|
\]
be the question denominator, \(N_C\) the number of exactly linked conditions, \(N_P\) the number of protocol-final conditions, and \(N_R\) the number with at least one observed redemption by \(\tstar\).  These counts answer different questions and are not collapsed into one rate.

The mapping report must distinguish:

\begin{itemize}
  \item exact one-to-one mapping;
  \item exact one-to-many or many-to-one protocol relation;
  \item no preparation event in the target stream;
  \item event present but upstream question linkage unavailable;
  \item conflicting identity or malformed event;
  \item right-censored downstream state.
\end{itemize}

\subsection{Protocol-finality path estimands}

For conditions with compatible upstream and downstream clocks, the primary intervals are
\begin{align}
  L_c^{OP} &= \tresolve(c)-\toracle(c),\\
  L_c^{OA,C} &= \tadapter(c)-\toracle(c),\\
  D_c^{PA} &= \tadapter(c)-\tresolve(c).
\end{align}

The signed adapter--protocol relation is reported with an order class:

\begin{enumerate}
  \item protocol event before adapter-terminal event in the same transaction;
  \item adapter-terminal event before protocol event in the same transaction;
  \item same block, different transactions;
  \item later block;
  \item one endpoint unavailable or interval-observed.
\end{enumerate}

A zero duration is not interpreted as simultaneous internal execution.  It means the available public clock has zero separation while exact event order is preserved where possible.

\subsection{Payout-vector taxonomy}

Let \(\mathcal P_c\) be the raw ordered numerator pattern and \(\pivec_c\) the normalized payout vector.  The taxonomy records both.  For binary conditions the principal normalized classes are:

\begin{align*}
  \text{YES} &: (1,0),\\
  \text{NO} &: (0,1),\\
  \text{NEUTRAL} &: \left(\tfrac12,\tfrac12\right),\\
  \text{OTHER VALID} &: \pivec_c\text{ valid but outside the canonical three},\\
  \text{MALFORMED/CONFLICTING} &: \text{protocol or decoding gate fails}.
\end{align*}

Scale-equivalent raw numerators remain distinct protocol encodings but share one normalized economic class.  Slot count, numerator order, and denominator remain visible in descriptive tables.

\subsection{First-redemption endpoints}

For any-redemption survival, define
\[
  T_c^{R}=\min_{j\in\Redemptions_c}\{t_{cj}-\tresolve(c)\}
\]
when a qualifying event exists.  Otherwise the observation is censored at
\[
  C_c=\tstar-\tresolve(c).
\]
The observed pair is
\[
  \widetilde T_c^{R}=\min(T_c^{R},C_c),
  \qquad
  \delta_c^{R}=\ind{T_c^{R}\leq C_c}.
\]

For positive-payout redemption, define \(T_c^{R+}\) analogously using only events with strictly positive paid collateral.  A prior zero-payout event does not remove the condition from the positive-redemption risk set.

\begin{figure}[htbp]
\centering
\resizebox{0.94\textwidth}{!}{%
\begin{tikzpicture}[
  every node/.style={font=\footnotesize},
  arr/.style={-{Latex[length=2mm]},thick},
  tick/.style={draw=black!70,thick}
]
\draw[arr] (0,0) -- (14.8,0) node[right] {time after protocol finality};
\draw[tick] (1,-0.12)--(1,0.12); \node[above=2mm] at (1,0) {$\tresolve$};
\draw[tick] (5,-0.12)--(5,0.12); \node[above=2mm,align=center] at (5,0) {zero-payout\\redemption};
\draw[tick] (9,-0.12)--(9,0.12); \node[above=2mm,align=center] at (9,0) {positive-payout\\redemption};
\draw[tick] (13,-0.12)--(13,0.12); \node[above=2mm] at (13,0) {$T^\star$};
\draw[seriesblue,very thick] (1,-0.55)--(9,-0.55); \node[left,font=\scriptsize] at (1,-0.55) {positive endpoint};
\draw[black!55,very thick] (1,-1.05)--(5,-1.05); \node[left,font=\scriptsize] at (1,-1.05) {any redemption};
\draw[dashed,black!65,very thick] (1,-1.55)--(13,-1.55); \node[left,font=\scriptsize] at (1,-1.55) {right-censored example};
\node[font=\scriptsize,align=center] at (7,-2.15) {The risk set begins at protocol finality, not at market creation or oracle request.};
\end{tikzpicture}%
}
\caption{Risk-set and endpoint construction.  A zero-payout redemption ends the ``any redemption'' endpoint, but follow-up for the positive-payout endpoint continues until a positive event or the snapshot.  Conditions without a qualifying event by the fixed snapshot are right-censored.}
\label{fig:censoring_riskset}
\end{figure}
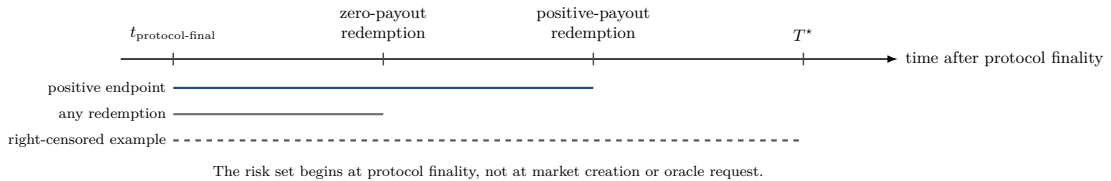

\subsection{Ordered redemption histories}

The condition-level holder state follows
\[
  0=\text{no redemption},
  \qquad
  1=\text{zero-only observed},
  \qquad
  2=\text{positive observed}.
\]
Allowed transitions are
\[
  0\rightarrow1,
  \qquad
  0\rightarrow2,
  \qquad
  1\rightarrow2.
\]
State 1 is not a terminal competing cause for state 2.  Aalen--Johansen transition probabilities can summarize the ordered multistate process, while Kaplan--Meier curves estimate time to any redemption and time to first positive redemption separately.

\subsection{Observed payout realization}

The identified condition-level cumulative payout is \(A_c(u)\) from \Cref{eq:observed_payout_curve}.  Aggregate observed payout can be reported as
\[
  A(u)=\sum_{c} A_c(u)
\]
over a stated population.  Acceptable normalizations include:

\begin{itemize}
  \item per protocol-final condition;
  \item per condition with any observed redemption;
  \item by payout-vector class;
  \item by adapter route or metadata stratum;
  \item by observed total payout over the complete follow-up window.
\end{itemize}

Each normalization is labelled.  None is described as entitlement completion unless \(W_c\) is exact and \(\valuegate\) passes.

\subsection{State-occupation and capital-time estimands}

For a state \(s\), the empirical occupation is \(\occupancy_{cs}\) from \Cref{eq:condition_state_occupation}.  Relevant states include:

\begin{itemize}
  \item oracle-final but not protocol-final;
  \item protocol-final with no observed redemption;
  \item zero-only observed;
  \item positive observed.
\end{itemize}

Occupation summaries can be combined with a registered capital-service rate \(\kappa_c^{\mathrm{cap}}(t)\) to form
\begin{equation}
  C_c^{\mathrm{cap}}(a,b)=\int_a^b \kappa_c^{\mathrm{cap}}(t)\,dt,
  \label{eq:capital_time}
\end{equation}
which measures capital-time rather than assuming one flat delay has the same economic value across conditions.

\subsection{Stratification}

Primary descriptive strata are fixed before results:

\begin{itemize}
  \item adapter family and route;
  \item payout-vector class;
  \item ordinary versus disputed/reset/Data Verification Mechanism (DVM) path inherited from Part~I;
  \item event class where exact metadata exists;
  \item clarification and rule-update history;
  \item same-transaction, same-block, and later-block downstream path;
  \item observation grade and metadata-coverage cohort.
\end{itemize}

Metadata-conditioned results display the unmatched residual and adapter composition.  A modern-adapter, high-coverage subset is not generalized to the full cohort without evidence.

\subsection{Uncertainty and inference}

The cohort is a frozen historical population rather than a simple random sample from future Polymarket conditions.  The primary outputs are exact counts, empirical distributions, and right-censored population summaries.  Confidence intervals, when reported, describe estimator uncertainty under stated event-history assumptions and are not automatic superpopulation guarantees.

Duration distributions are heavy-tailed and can contain point masses at zero.  Reports therefore use medians, quantiles, restricted means over prespecified horizons, numbers at risk, and full survival coordinates rather than relying only on arithmetic means.

\subsection{Negative controls}

The empirical build includes negative controls designed to detect false structure:

\begin{enumerate}
  \item shift protocol-finality times across unrelated conditions while preserving marginal distributions;
  \item replace exact condition mapping with a forbidden title-similarity candidate relation and verify that it is excluded from canonical results;
  \item randomize downstream event order within blocks to test dependence on exact log ordering;
  \item compare source-observation counts with canonical event counts to detect transport duplication;
  \item compute a creation-time redemption curve to demonstrate the bias from an incorrect risk-set origin;
  \item treat zero-payout redemption as terminal for positive redemption and quantify the resulting endpoint distortion.
\end{enumerate}

Negative controls are diagnostics, not alternative primary estimators.

\subsection{Claim hierarchy}

The claim hierarchy is cumulative:

\begin{enumerate}
  \item exact event existence and identity;
  \item exact or graded state linkage;
  \item protocol-finality path and duration;
  \item right-censored holder-action timing;
  \item observed payout realization;
  \item entitlement-completion claim, only after the value denominator passes.
\end{enumerate}

A failure at one level limits all stronger downstream claims.

\section{The Economics of Protocol Finality and Observed Redemption}
\label{sec:economics_gap}

The state sequence identified in the preceding sections has an economic interpretation, but the interpretation depends on which actor controls the remaining transition.  Before protocol finality, the mechanism has not yet recorded the enforceable payout rule.  After protocol finality, the claim can be technically redeemable while the holder has not yet exercised the action.  These intervals can both create capital time, but they do not represent the same risk.

\subsection{Cash-equivalent value}

Let \(V_c(t)\) be the executable or accounting value of one unit of the winning claim at time \(t\), expressed in collateral units, and let \(\pi_c^{\mathrm{win}}\) be its final per-unit winning payout under the recorded vector. Define the observed cash-equivalent wedge
\begin{equation}
  \omega_c(t)=\pi_c^{\mathrm{win}}-V_c(t).
  \label{eq:cash_equivalent_wedge}
\end{equation}

Before protocol finality, \(\omega_c(t)\) can include adjudication uncertainty, downstream execution uncertainty, liquidity cost, and discounting.  After protocol finality, adjudication uncertainty should be absent under the ordinary path, but transaction cost, custody delay, batching, collateral access, and holder optionality can remain.  The same numerical wedge can therefore have a different decomposition across states.

A positive post-finality wedge is not automatically protocol failure.  It can be rational compensation for the cost of realizing collateral or an artifact of a stale/illiquid quoted price.  The empirical design treats the wedge as an outcome to explain, not a target that must mechanically be forced to zero.

\subsection{Known-payout carry benchmark}

Suppose protocol finality has fixed a claim payout \(\pi_c\), and a holder expects redemption after random delay \(\Delta_c\). Under a simple collateral carry rate \(r\), transaction cost \(k_c\), and no further default risk, a benchmark present value is
\begin{equation}
  V_c^{\mathrm{bench}}
  =
  \E\!\left[\pi_c^{\mathrm{win}} e^{-r\Delta_c}\right]-k_c.
  \label{eq:known_payout_benchmark}
\end{equation}

This is a benchmark, not a structural pricing model.  It clarifies that a known payout can still have a positive discount when realization is delayed or costly.  More elaborate models can include stochastic borrowing rates, custody risk, segmentation, account-level gas sponsorship, and endogenous redemption timing.

\subsection{Mechanism delay versus holder delay}

Define the mechanism interval
\[
  \Delta_c^{\mathrm{mech}}
  =
  \tresolve(c)-\toracle(c)
\]
and the holder interval
\[
  \Delta_c^{\mathrm{holder}}
  =
  \tredeem(c)-\tresolve(c).
\]

The first interval is controlled by the adapter/protocol path.  The second includes holder or custodian choice after the payout rule exists.  A capital provider can be exposed to one, both, or neither interval depending on the instrument architecture.  For example, a loan extinguished at protocol finality has no principal exposure to later holder delay, while a service promising immediate collateral delivery can retain an operational obligation until redemption.

\subsection{State-occupation value}

A scalar duration can conceal differences in state occupation.  Let \(k_s(t)\) be the per-unit cost rate associated with state \(s\).  The total state-conditioned cost is
\begin{equation}
  C_c(a,b)
  =
  \sum_s
  \int_a^b
  k_s(t)\ind{X_c(t)=s}\,dt.
  \label{eq:state_cost}
\end{equation}

Two conditions can have equal protocol-to-redemption duration but different cost because one spends most of the interval in a zero-only state, one requires multiple redemption transactions, or one is handled through a custodian.  This motivates reporting occupation and route rather than only a single gap.

\subsection{Leverage amplification}

For a financed position of notional \(N\) and trader equity \(E\), a finality discount \(d\) reduces marked equity by approximately \(Nd\) before fees and nonlinear execution effects.  The same per-unit discount therefore has a larger equity effect at higher leverage.  But the economic response depends on the state:

\begin{itemize}
  \item before protocol finality, the discount can reflect unresolved mechanism risk;
  \item after protocol finality, it can reflect realization cost or optional delay;
  \item after debt extinction, it can belong entirely to the holder rather than the lender;
  \item under a guaranteed conversion service, part of the wedge can belong to the operator or reserve.
\end{itemize}

The paper does not infer a universal margin add-on from the observed wedge.  It supplies the state and duration variables needed by Paper~6 to test such rules.

\subsection{Observable economic estimands}

The empirical release can support the following observables without an entitlement denominator:

\paragraph{Oracle-to-protocol delay.}
The distribution of \(L^{OP}\) measures downstream mechanism delay after oracle finality.

\paragraph{Adapter--protocol order.}
The signed \(D^{PA}\) and its exact order class measure the emitted sequence without imposing an artificial direction.

\paragraph{First-redemption survival.}
\(\km(u)\) estimates the share of protocol-final conditions with no observed redemption by horizon \(u\), under the stated censoring rules.

\paragraph{Positive-redemption survival.}
The positive endpoint separates zero-payout cleanup from observed positive collateral realization.

\paragraph{Observed payout curve.}
\(A_c(u)\) and its aggregates measure collateral paid through observed events.

\paragraph{State-price relation.}
Where executable prices exist, the relation between the cash-equivalent wedge and expected remaining protocol or holder duration can be studied after controlling for spread, staleness, event class, route, and payout class.  Such an analysis remains associational.

\subsection{Falsification and interpretation}

A finality-pricing interpretation is weakened if, within high-grade cohorts, the observed wedge has no stable relation to measured remaining duration, state occupation, route, or realization delay after controlling for liquidity and market characteristics.  A null result would be informative: the delay may be too short to price, holders may treat redemption as operationally negligible, or measurement noise may dominate.

A null pricing result would not collapse the state-machine contribution.  Funding, collateral, custody, and conversion feasibility can remain state-dependent even when no visible spot discount appears.  Conversely, a strong observed discount is not by itself proof of oracle or protocol risk.  The claim is conditional on coverage, timing, price, and value gates.

\section{Registered Empirical Design and Analysis Sequence}
\label{sec:empirical_contract}

The analysis is ordered because every downstream result depends on an earlier identity, coverage, or population predicate.  The sequence below is part of the scientific design rather than a software execution diary.  A later table can refine an earlier description, but it cannot repair a failed event-family, coverage, identity, or linkage gate.

\subsection{Target-stream integrity and coverage}

The first analysis verifies the ABI-derived event registry, exact block-range coverage, canonical identities, source duplication, decoding, and deterministic rebuild for \eventcode{ConditionPreparation}, \eventcode{ConditionResolution}, and \eventcode{PayoutRedemption}.  Its outputs distinguish source observations from canonical events and classify every anomaly.  Population rates are unavailable until this layer passes.

\subsection{Cohort linkage and attrition}

The second analysis maps the frozen adapter-question cohort to condition identities.  It reports exact one-to-one, one-to-many, many-to-one, unmatched, ambiguous, and inapplicable states.  Title similarity, fuzzy matching, and nearest-time candidates appear only as negative controls and never enter the accepted relation.  Full-contract event counts remain separate from cohort denominators.

\subsection{Protocol-finality paths}

The third analysis links Oracle-final, adapter-terminal, and condition-resolution events.  It distinguishes same-transaction protocol-before-adapter, same-transaction adapter-before-protocol, same-block different-transaction, later-block, interval-observed, unavailable-upstream, and right-censored paths.  Duration summaries are restricted to compatible clock grades, while exact secondary event order is retained even when wall-clock duration is zero.

\subsection{Payout-vector taxonomy}

The fourth analysis preserves ordered payout numerators, slot counts, denominators, normalized classes, and malformed or conflicting attrition.  Scale-equivalent encodings are separated from economically distinct payout classes.  Composition by adapter and adjudication route is reported without assuming that a human-readable YES/NO label determines slot order.

\subsection{Redemption histories}

The fifth analysis constructs first-any, first-positive, and ordered zero-to-positive histories.  It reports event and risk-set counts, administrative censoring, Kaplan--Meier coordinates, multistate transitions, numbers at risk, quantiles, and restricted means where identified.  A zero-payout redemption terminates the any-redemption endpoint but not the positive-redemption endpoint.  Creation-time placebo and false competing-risk specifications remain negative controls.

\subsection{Observed payout realization}

The sixth analysis aggregates exact collateral paid by observed redemption events and reports explicitly named normalization bases.  Absolute observed payout is identified from the event ledger.  Entitlement-completion ratios remain unavailable unless the separate balance-consistent denominator passes \(\valuegate\).

\subsection{State variables for related analyses}

The analysis identifies state variables relevant to Papers~6, 8, and~9: condition identity, payout class, protocol-finality path, exact event-order keys, first-redemption and first-positive timing, censoring state, and route metadata.  Counterfactual margin, cross-venue comparison, and multi-leg replay remain outside this paper.

\subsection{Version-locked hypotheses}

The empirical analysis evaluates the following hypotheses without retroactive replacement.

\begin{hypothesis}[Downstream separation]
\label{hyp:downstream-separation}
Among conditions with compatible observed clocks, a non-zero share has a strictly positive or interval-qualified Oracle-to-protocol finality gap; exact event-identity non-equivalence is treated as architecture, not as the empirical test.
\end{hypothesis}

\begin{hypothesis}[Atomic ordinary path]
\label{hyp:atomic-ordinary-path}
Among standard ordinary-path conditions, protocol-resolution and adapter-terminal events frequently occur in one transaction, while exact log order remains non-degenerate.
\end{hypothesis}

\begin{hypothesis}[Payout concentration]
\label{hyp:payout-concentration}
Canonical YES and NO payout classes dominate valid binary condition resolutions, with neutral and other valid classes forming a smaller but non-zero tail.
\end{hypothesis}

\begin{hypothesis}[Redemption censoring]
\label{hyp:redemption-censoring}
A material share of protocol-final conditions remains without an observed redemption at the fixed follow-up snapshot.
\end{hypothesis}

\begin{hypothesis}[Positive endpoint delay]
\label{hyp:positive-endpoint-delay}
A non-zero share of resolved conditions satisfies \(T_c^{R+}>T_c^R\); the release reports the share and the distribution of the positive gap rather than testing the pointwise ordering itself.
\end{hypothesis}

\begin{hypothesis}[Route heterogeneity]
\label{hyp:route-heterogeneity}
Protocol-finality and redemption timing differ across adapter and upstream adjudication routes.
\end{hypothesis}

\paragraph{Analytic value boundary.} Theorem~\ref{thm:entitlement_nonidentification} establishes that redemption events alone do not identify entitlement completion. The empirical analysis therefore reports the share of outputs affected by the missing entitlement denominator and keeps G-VALUE blocked unless a separate balance-complete reconstruction is supplied.

A rejected hypothesis remains part of the analysis.  A blocked hypothesis is reported as blocked rather than evaluated on an opportunistic denominator.

\subsection{Result order and claim gating}

Results appear in the following order: target-stream integrity; cohort mapping and attrition; protocol-finality paths; payout taxonomy; first-any and first-positive redemption; ordered zero-to-positive histories; observed payout realization; the value-gate disposition; and the state-variable boundary.  Each table names its population or risk set and its cumulative gate state.

\begin{table}[htbp]
\centering
\caption{Frozen result families and minimum denominator reporting.}
\label{tab:results_contract}
\begin{tabularx}{\textwidth}{L{0.24\textwidth}L{0.29\textwidth}Y}
\toprule
Result family & Required denominator / risk set & Minimum output\\
\midrule
Population flow & frozen questions and exact linked conditions & counts, rates, and exclusion reasons\\
Protocol-finality paths & exact linked upstream terminal states & same-transaction, same-block, delayed, censored/ineligible\\
Payout vectors & valid condition-resolution events & raw patterns, normalized classes, malformed attrition\\
First redemption & protocol-final conditions with follow-up & event count, censoring, quantiles, survival, numbers at risk\\
Positive redemption & protocol-final conditions; prior zero-payout events retained in the history & first-positive survival and ordered zero-to-positive transition counts\\
Observed realization & canonical redemption events and explicitly named normalization base & absolute collateral and labelled observed-only curves\\
Value completion & exact entitlement denominator & reported only if G-VALUE passes\\
\bottomrule
\end{tabularx}
\end{table}

\subsection{Source-generated result layer}

Reported numerical tables and vector figures are generated from frozen inputs.  Their manifests record input hashes, analytical revision, exact denominator, clock origin, gate state, output schema, and deterministic rebuild status.  Inherited Part I constants remain documentary inputs rather than regenerated Part II outputs.

\section{Empirical Results}
\label{sec:results}

\subsection{Population linkage and lifecycle composition}
The exact protocol bridge preserves the distinction between adapter-question instances, unique question IDs, link rows, and conditions. One question ID appears in two valid adapter-question instances, and one mapped question ID links to two CTF conditions; neither multiplicity is collapsed by convenience. \Cref{tab:results_cohort_attrition} reports the registered units.

\begin{table}[htbp]
\centering
\caption{Frozen cohort, identity units, and exact question-to-condition linkage. Counts with different units are shown separately rather than forced into one additive funnel.}
\label{tab:results_cohort_attrition}
\small
\begin{tabularx}{\textwidth}{@{}lY@{}}
\toprule
Registered unit & Count \\
\midrule
Initialized adapter-question instances & 185,550 \\
Unique question IDs & 185,549 \\
Exact one-to-one mapped question IDs & 108,636 \\
One-to-many mapped question IDs & 1 \\
Unmatched unique question IDs & 76,912 \\
Exact question-to-condition link rows & 108,639 \\
Unique linked conditions & 108,638 \\
\bottomrule
\end{tabularx}
\end{table}

Among the 108,638 exact-linked conditions, 99,283 have an observed protocol-resolution event by the fixed snapshot at Polygon block 90,114,204 (2026-07-12T17:11:41Z). The remaining lifecycle states are shown in \Cref{tab:results_lifecycle}. Redemption without an observed resolution is retained as its own observational class rather than silently reassigned.

\begin{table}[htbp]
\centering
\caption{Observed lifecycle states in the exact-linked condition cohort. These are event-observation states, not claims about beneficial ownership or total economic entitlement.}
\label{tab:results_lifecycle}
\small
\begin{tabularx}{\textwidth}{@{}Yrr@{}}
\toprule
State & Conditions & Share \\
\midrule
Positive redemption observed & 91,817 & 84.52\% \\
Zero-only first redemption, no later positive & 341 & 0.31\% \\
Resolved, no redemption by snapshot & 7,125 & 6.56\% \\
Redemption observed without observed resolution & 3,669 & 3.38\% \\
No observed resolution or redemption & 5,686 & 5.23\% \\
\midrule
Total exact-linked conditions & 108,638 & 100.00\% \\
\bottomrule
\end{tabularx}
\end{table}

\subsection{Payout-vector taxonomy}
The corrected payout-vector taxonomy reconciles exactly to the observed resolution counts: 99,283 resolved exact-linked conditions and 1,857,116 resolved conditions in the total target-event union. Binary vectors dominate, while neutral and other valid vectors form a small but non-zero tail. \Cref{tab:results_payout_taxonomy} reports population-qualified counts.

\begin{table}[htbp]
\centering
\caption{Payout-vector taxonomy by registered population. Orientation follows the ordered protocol vector; no human-readable YES/NO label is used to infer slot order.}
\label{tab:results_payout_taxonomy}
\scriptsize
\begin{tabular}{lrrr}
\toprule
Taxonomy & Exact-linked & Full-CTF-only & Total target union \\
\midrule
Canonical binary $(0,1)$ & 53,847 & 1,041,734 & 1,095,581 \\
Canonical binary $(1,0)$ & 45,024 & 686,769 & 731,793 \\
Fifty-fifty & 410 & 29,199 & 29,609 \\
Other valid vector & 2 & 131 & 133 \\
No observed resolution & 9,355 & 177,718 & 187,073 \\
\bottomrule
\end{tabular}
\end{table}

\subsection{Protocol-finality paths}
Cross-contract finality is deliberately conservative. The frozen evidence does not support a universal same-transaction ordering claim, and an identity join is not promoted to causal ordering. The condition-level classification in \Cref{tab:results_protocol_paths} therefore distinguishes interval-qualified, multiple-candidate, incompatible, and unevaluable paths.

\begin{table}[htbp]
\centering
\caption{Conservative condition-level protocol-finality path classification for the exact-linked cohort. The classification does not infer same-transaction order when exact cross-contract transaction/log identity is unavailable.}
\label{tab:results_protocol_paths}
\small
\begin{tabularx}{\textwidth}{@{}Yr@{}}
\toprule
Path class & Conditions \\
\midrule
Interval-qualified terminal generation & 823 \\
Multiple candidate generations & 48 \\
No compatible terminal generation in frozen evidence & 91,638 \\
Right-censored or otherwise unevaluable & 16,129 \\
\midrule
Total & 108,638 \\
\bottomrule
\end{tabularx}
\end{table}

The same-transaction Oracle--adapter--CTF and same-block/same-transaction redemption estimands are explicitly not evaluable from the transferred frozen state layer because the required cross-contract transaction/log identity is not retained at publication grade. No ordering is imputed.

\subsection{Observed redemption timing}
Redemption follow-up uses each condition's first observed protocol resolution as time zero and right-censors at the exact snapshot. In the exact-linked resolved risk set of 99,283 conditions, 92,158 have an observed redemption of any amount and 91,817 have an observed positive-payout redemption. The Kaplan--Meier medians are 182 and 200 seconds respectively; \Cref{tab:results_survival} reports the risk-set accounting.

\begin{table}[htbp]
\centering
\caption{Condition-specific redemption survival summaries for the exact-linked resolved cohort. Time zero is first observed protocol resolution; unresolved conditions are outside this risk set.}
\label{tab:results_survival}
\small
\begin{tabularx}{\textwidth}{@{}L{3.2cm}rrrrrY@{}}
\toprule
Endpoint & Risk set & Events & Censored & Median (s) & P25 (s) & P75 (s) \\
\midrule
First any redemption & 99,283 & 92,158 & 7,125 & 182 & 52 & 544 \\
First positive redemption & 99,283 & 91,817 & 7,466 & 200 & 58 & 720 \\
\bottomrule
\end{tabularx}
\end{table}

Available follow-up is heterogeneous because conditions resolve at different dates. For exact-linked resolved conditions, the 25th, 50th, and 75th percentiles of available follow-up are 21,074,649, 25,254,353, and 33,617,219 seconds, and the maximum is 142,006,845 seconds. These values describe observation exposure, not expected redemption delay.

\subsection{Entitlement boundary}
Observed redemption events identify actions and paid collateral. They do not identify the fraction of all winning entitlement redeemed because the release does not contain a full balance-complete ERC-1155 entitlement ledger. Accordingly, G-VALUE remains unavailable and the paper reports no entitlement-completion rate, outstanding winning-token supply, or unredeemed collateral value. This is a scope boundary of the event-timing study rather than a failure of the event-coverage or censoring results.

\section{Implications for Leveraged Event Markets}
\label{sec:design_implications}

The paper is an empirical protocol study, not a complete leveraged-product specification.  Its contribution to mechanism design is to identify the exact states and clocks that a financed overlay must distinguish.

\subsection{Cash finality and debt maturity}

A lender or margin engine should attach discharge to the state that actually satisfies the obligation.  Oracle finality can be sufficient for a prediction about direction, but not necessarily for a protocol payout.  Protocol finality establishes redeemability, but not observed collateral in a holder account.  Redemption can be optional holder behavior rather than an operator obligation.

\begin{designprinciple}[Finality-state accounting]
\label{dp:finality-state-accounting}
A leveraged event engine should attach each obligation to the exact state required for discharge: oracle-final, adapter-terminal, protocol-final, delivered claim, or redeemed collateral.  A weaker upstream state must not be used as a substitute.
\end{designprinciple}

This principle does not require every liability to remain open until holder redemption.  Debt can be extinguished earlier through a verified sale, conversion, reserve, or other enforceable transition.  The mechanism must state which asset and event satisfy the obligation.

\subsection{Perpetual-to-conditional-token conversion}

Paper~6 develops conversion from a synthetic event-linked position into fully collateralized conditional-token exposure.  Part~II supplies three load-bearing inputs:

\begin{enumerate}
  \item exact condition and outcome-slot structure;
  \item the protocol-final payout state and claim identity;
  \item the observed distinction between claim delivery, redeemability, and redemption.
\end{enumerate}

A conversion rule triggered before exact condition identity or enforceable delivery can replace one unresolved obligation with another.  Conversion after debt extinction can instead leave the user with a fully funded residual claim whose later redemption belongs to the holder rather than the lender.

\subsection{State-triggered margin and finality carry}

Protocol-finality gaps create carrying costs even when the final economic direction is known.  A control rule can condition on
\[
  L^{OP},\quad L^{OA},\quad D^{PA},\quad L^{PR},
\]
but these intervals belong to different actors and risks.  Funding before protocol finality can compensate an unresolved mechanism state.  A charge after protocol finality may instead be a custody, realization, or capital-service charge.  Calling both ``basis funding'' obscures the change in economic function.

Part~II supplies empirical distributions; Paper~6 evaluates candidate margin, funding-freeze, auction, and conversion rules.  This paper does not prescribe one calibration.

\subsection{Asynchronous multi-leg instruments}

An event spread has a protocol-finality and redemption path for each leg.  When leg \(A\) becomes protocol-final before leg \(B\), the spread becomes a deterministic component plus a residual single-leg exposure.  Paper~9 formalizes this state and its collateral consequences.  Part~II supplies leg-level condition-resolution and holder-action clocks.

A common event title or oracle-final timestamp cannot substitute for exact condition and payout identities.  Multi-leg conversion also requires aggregate claim capacity and cannot reuse the same complete-set inventory across positions.

\subsection{Risk-transfer boundary}

Observed redemption delay should not automatically be booked as protocol risk.  The design must separate:

\begin{itemize}
  \item inability to record a valid payout rule;
  \item delay between oracle and protocol finality;
  \item delay between protocol finality and claim delivery;
  \item holder delay after redeemability;
  \item missing entitlement information that prevents completion measurement.
\end{itemize}

This separation avoids charging the protocol for optional holder behavior while retaining the capital-time experienced by a service that guarantees realization.

\subsection{Outcome-sensitive exposure}

A zero-payout redemption can be operational cleanup, while a positive-payout redemption releases collateral.  A leveraged engine should therefore avoid interpreting any redemption event as a uniform reduction in economic exposure.  The relevant quantity is the exact payout, claim inventory, and obligation being discharged.

Similarly, a 50/50 or other valid payout vector changes the identity and quantity of claims that carry value.  Margin and conversion rules must ingest the actual ordered payout vector rather than a hard-coded YES/NO assumption.

\subsection{Design boundary}

The empirical clocks do not establish a welfare-optimal leverage cap, conversion rule, or reserve.  Those choices require execution cost, strategic response, legal enforceability, shared capacity, and participant behavior.  The paper's design output is a verified state surface and claim boundary, not a deployed product recommendation.

\section{Reproducibility and Data Availability}
\label{sec:reproducibility}

The companion materials are organized by adapter, oracle, rules and metadata, protocol target
events, redemption, analysis, coverage, provenance, and manifests.  The protocol layer contains
only ABI-derived events with complete canonical envelopes.  Provider observations remain
provenance and do not inflate canonical event counts.

\subsection{Event and table lineage}

The required lineage is
\[
\begin{aligned}
  \text{source observation}
  &\rightarrow \text{canonical event}
  \rightarrow \text{exact relation}\\
  &\rightarrow \text{analytical table or figure}.
\end{aligned}
\]

Each reported table and figure has a source-manifest entry.  A claim can therefore be traced to
the target-event registry, range ledger, source envelopes, canonical identity, linkage relation,
and analytical code.

\subsection{Deterministic rebuild}

A network-disabled rebuild of the reported tables and figures uses the frozen target-event
dataset and accepted upstream Part~I relations.  The rebuild compares schemas, key sets, semantic
hashes, range coverage, linkage, attrition, survival tables, and figure-source hashes.  Byte
identity is required where deterministic file metadata permits it; semantic identity is required
everywhere.

The check has two components:
\begin{enumerate}
  \item a source-to-canonical verifier for event identity, semantic payload, provenance,
        duplicate classification, and range coverage; and
  \item an analytical verifier for mapping, payout, timing, survival, and observed-realization
        relations.
\end{enumerate}

A content-addressed inventory check does not replace semantic verification.  Both checks are
recorded in the reproducibility materials.

\subsection{Data availability and privacy boundary}

The accompanying public materials describe the analytical tables, schemas, provenance, known
limitations, and citation metadata needed to assess the reported claims.  Redemption addresses
are public pseudonymous identifiers, but the reported material does not provide ranked holder
lists by default.  Tables are aggregated unless row-level disclosure is necessary for a stated
reproducibility purpose and is consistent with minimum-necessary disclosure.  Beneficial
ownership, intent, institutional identity, and legal status are never inferred from an address
alone.

\section{Limitations}
\label{sec:limitations}

The paper has sixteen principal limitations.

\paragraph{Contract and venue scope.}
The protocol layer is Polymarket-, chain-, contract-, and version-specific.  Exact event coverage and ABI decoding establish what the pinned contracts recorded; they do not establish legal finality in every jurisdiction or portability to another venue.

\paragraph{Frozen historical cohort.}
The question cohort ends at block \(79{,}721{,}080\), and follow-up ends at block \(90{,}114{,}204\), whose exact canonical timestamp is 2026-07-12T17:11:41Z.  Later questions and events are outside every rate.  The manuscript date is not the data cutoff.

\paragraph{Question and condition identity.}
The primary upstream unit is an adapter address and question-identifier pair.  It is not automatically one economic market across migrations, wrappers, duplicate listings, or semantically related questions.  Condition mapping is a protocol relation rather than a universal reconciliation of economic identity.

\paragraph{Condition preparation.}
Preparation is a technical registry event.  It does not prove market activity, evidence sufficiency, contractual decidability, or the later use of one specific resolution path.

\paragraph{External semantic clocks.}
Part~I reports that external-source publication and contractual-decidability clocks are unmeasured at population scale.  Part~II therefore does not reinterpret request, oracle, protocol, or redemption times as substitutes for those semantic clocks.

\paragraph{Protocol finality.}
A condition-resolution event is protocol finality under the observed contract.  It is not necessarily the final movement of collateral into a holder's externally controlled account, and it does not prove that every interface or custodian exposes immediate redemption.

\paragraph{Adapter--protocol order.}
The publication-grade frozen state layer does not retain enough cross-contract transaction/log identity to evaluate a universal same-transaction Oracle--adapter--CTF order. Where exact secondary order is unavailable, the paper reports interval-qualified or unevaluable paths rather than inferring an order from equal timestamps or shared identifiers.

\paragraph{Redemption optionality.}
No observed redemption by the snapshot is right censoring, not proof of protocol failure, abandonment, or economic loss.  A holder can choose not to redeem a technically available claim.

\paragraph{Address identity.}
Address-level counts are not beneficial-owner counts.  Multiple addresses can belong to one actor, one address can act for many users, and relayers, custodians, or account-abstraction systems can mediate actions.

\paragraph{Entitlement denominator.}
The event ledger does not generally identify total outstanding winning-token entitlement.  The paper therefore does not estimate true entitlement-completion rates or total unredeemed value unless a separate balance reconstruction passes \(\valuegate\).

\paragraph{Wrappers and collateral routes.}
Negative-risk wrappers, collateral wrappers, migrated adapters, bridges, fees, burns, and off-chain netting can affect balances and realization without appearing as a simple binary CTF flow.

\paragraph{Transport and historical labels.}
Provider transports and historical stream labels are not chain truth.  The publication build mitigates this through ABI-derived topics, full envelopes, range ledgers, canonical identities, receipt-based log-index recovery, provenance, and independent overlap checks.  Residual provider failure remains an operational limitation.

\paragraph{Censoring and selection.}
Protocol-final and redemption-observed conditions are selected.  Survival methods correct risk-set accounting under stated censoring assumptions but do not identify causal effects of route, payout class, metadata, or governance history.

\paragraph{Price and liquidity after finality.}
There is no universal post-protocol order book against which all conditions can be compared.  Some claims stop trading, become stale, or move through wrappers.  Cash-equivalent wedge analysis therefore requires a separate price and liquidity gate.

\paragraph{Leveraged-design boundary.}
State-triggered margin, claim conversion, auctions, reserves, and multi-leg controls are mechanism families studied in Papers~6 and~9.  No result here establishes a production-safe or welfare-optimal leveraged product.

\paragraph{Release boundary.}
Data disclosure boundary. The article reports aggregate lifecycle, payout, redemption, and censoring results; holder-level entitlement reconstruction and raw holder-level records are not distributed.

\paragraph{Payout-vector scope.}
Payout-vector frequencies are reported separately for exact-linked conditions and the full CTF-only population. Contract-wide counts are coverage and protocol-history results; they are not substituted for Polymarket cohort rates.

\section{Conclusion}
\label{sec:conclusion}

Resolution is not one event. On Polymarket, Oracle adjudication, adapter terminality, protocol payout recording, redeemability, and observed holder redemption are distinct states with distinct identities, clocks, actors, and claim boundaries.

The final exact bridge contains 108,638 linked conditions. A protocol-resolution event is observed for 99,283 of them by the exact snapshot at 2026-07-12T17:11:41Z. Among that resolved risk set, 92,158 conditions have an observed redemption of any amount and 91,817 have an observed positive-payout redemption; the corresponding condition-specific Kaplan--Meier medians are 182 and 200 seconds. These are event-timing results conditional on an observed protocol resolution, not estimates of eventual economic entitlement completion.

The corrected payout taxonomy is internally reconciled to protocol-resolution incidence. The exact-linked cohort contains 53,847 canonical $(0,1)$ payout vectors, 45,024 canonical $(1,0)$ vectors, 410 fifty-fifty vectors, two other valid vectors, and 9,355 conditions with no observed resolution. The full CTF universe is retained separately and never substituted for the frozen Polymarket cohort.

Cross-contract finality remains intentionally conservative. The frozen publication layer supports condition-level interval and candidate classifications but not a universal same-transaction Oracle--adapter--protocol ordering claim. This non-evaluability is itself a result of the evidence contract: missing exact cross-contract ordering keys are not repaired by timestamps or identity joins.

Finally, event-complete redemption data do not by themselves identify total winning entitlement. A broader balance-complete reconstruction could in principle supply that denominator, but it is not part of this release. The paper therefore makes no G-VALUE claim about the fraction of all entitlement redeemed, outstanding winning-token supply, or unredeemed collateral value.

The resulting protocol states and clocks are now suitable inputs to the downstream margin, cross-venue, and asynchronous-finality studies, subject to the same population, observation-grade, and claim boundaries fixed here.

\appendix
\section{Operational Entity, Event, State, and Clock Dictionary}
\label{app:dictionary}

\subsection{Entity hierarchy}

\begin{longtable}{@{}L{0.21\textwidth}L{0.27\textwidth}L{0.42\textwidth}@{}}
\caption{Canonical entities and keys.}\label{tab:entity_dictionary}\\
\toprule
Entity & Canonical key & Interpretation\\
\midrule
\endfirsthead
\toprule
Entity & Canonical key & Interpretation\\
\midrule
\endhead
Adapter-question & adapter address + question identifier & frozen upstream population unit inherited from Part~I\\
Oracle request generation & oracle address + request key + generation relation & one technical adjudication request within a question history\\
Condition & canonical condition identifier & CTF registry object with oracle, question identifier, and slot count\\
Source observation & accepted page/range + source row identity & provider occurrence; can repeat another observation\\
Canonical event & chain + contract + transaction hash + log index & one exact on-chain event\\
Redemption event & canonical event identity & one observed redemption action\\
Provenance group & content-addressed group ID & set-valued source-to-event relation for indistinguishable repeated observations\\
Analytical condition row & condition ID + frozen analysis revision & merged protocol-finality and holder-history state\\
\bottomrule
\end{longtable}

\subsection{Core event families}

The publication event registry is generated from the pinned ABI.  The load-bearing target families are:

\begin{itemize}
  \item \textbf{Condition preparation.} Registry relation among oracle, question identifier, slot count, and condition identifier.
  \item \textbf{Condition resolution.} Ordered payout numerators and protocol-finality event.
  \item \textbf{Payout redemption.} Redeemer, collateral, parent collection, condition, index sets, and paid collateral.
\end{itemize}

Transfer, split, merge, wrapper, and other known CTF events can support value reconstruction but do not enter a target-event count under a different label.

\subsection{Clock dictionary}

\begin{longtable}{@{}L{0.19\textwidth}L{0.28\textwidth}L{0.43\textwidth}@{}}
\caption{Protocol and holder clocks.}\label{tab:clock_dictionary}\\
\toprule
Clock & Exact or graded source & Interpretation\\
\midrule
\endfirsthead
\toprule
Clock & Exact or graded source & Interpretation\\
\midrule
\endhead
\(\tprepare\) & condition-preparation event & registry time, not finality\\
\(\toracle\) & accepted Part~I endpoint & irreversible consumable oracle result; exact or interval-observed\\
\(\tconsume\) & adapter call path where identifiable & downstream consumption instant; not always separately emitted\\
\(\tresolve\) & first valid condition-resolution event & protocol payout rule recorded\\
\(\tadapter\) & exact adapter terminal event & observed adapter emission time\\
\(\tredeemable\) & contract-state implication of valid payout recording & technical availability; can coincide with protocol finality\\
\(\tredeem\) & earliest accepted redemption event & first observed holder action of any payout amount\\
\(\tredeempos\) & earliest strictly positive payout redemption & first observed positive collateral realization\\
\(\tstar\) & fixed snapshot block timestamp & administrative censoring boundary\\
\bottomrule
\end{longtable}

\subsection{Condition-level states}

\begin{itemize}
  \item \textbf{Unprepared.} No accepted preparation event in the target stream.
  \item \textbf{Prepared.} Condition registry relation exists; no valid payout rule recorded.
  \item \textbf{ProtocolFinal.} Valid payout numerators recorded.
  \item \textbf{Redeemable.} Positions can be redeemed under the recorded rule, subject to deployed contract/collateral path.
  \item \textbf{NoRedemption.} No accepted redemption observed by current time.
  \item \textbf{ZeroOnlyObserved.} At least one zero-payout redemption and no positive redemption observed yet.
  \item \textbf{PositiveObserved.} At least one strictly positive redemption observed.
  \item \textbf{RightCensored.} Follow-up ends before the relevant holder endpoint is observed.
  \item \textbf{Conflict.} Load-bearing identity or semantic evidence disagrees.
\end{itemize}

\subsection{Ordering and atomicity}

Block timestamp is the primary wall-clock field.  Exact order uses block number and log index, with transaction hash and contract address as deterministic tie-break fields.  A nullable transaction index is retained for validation but is not required for canonical event identity.  Same-transaction events can have zero duration in seconds and a strict exact order.

\subsection{Observation grades}

\begin{description}[style=nextline,leftmargin=2.6cm]
\item[$\ObsExact$ --- Exact] Exact event time and order.
\item[$\ObsInterval$ --- Interval] Endpoint bounded by exact observations but not point-identified.
\item[$\ObsSnapshot$ --- Snapshot] Current or terminal representation observed without transition history.
\item[$\ObsProxy$ --- Proxy] Documented substitute with an explicit relationship to the target coordinate.
\item[$\ObsUnmeasured$ --- Unmeasured] Evidence does not identify the coordinate.
\item[$\ObsConflict$ --- Conflicting] Load-bearing evidence disagrees.
\end{description}

Grades attach to coordinates rather than entire conditions.  A condition can have exact protocol finality and unmeasured entitlement.

\section{Exact Mapping and Canonical Data Schema}
\label{app:mapping_schema}

\subsection{Question-to-condition relation}

The canonical relation has one row per frozen adapter-question and records:

\begin{itemize}
  \item adapter-question identifier and initialization evidence;
  \item exact question identifier used by the condition registry;
  \item condition identifier;
  \item oracle address and outcome-slot count;
  \item condition-preparation canonical identity;
  \item mapping method and evidence hashes;
  \item ambiguity, multiplicity, and exclusion status.
\end{itemize}

Allowed mapping methods are exact protocol identities only.  The release retains unmatched questions and every alternative exact relation rather than selecting by title or time proximity.

\subsection{Source-observation schema}

One source-observation row preserves:

\begin{itemize}
  \item chain and contract;
  \item block number/hash/timestamp;
  \item transaction hash;
  \item raw nullable transaction index and status;
  \item exact log index;
  \item exact \code{removed} value;
  \item complete ordered topics and complete data;
  \item accepted page/range and provider provenance;
  \item raw response/page hashes;
  \item ABI-derived event type and decoder revision;
  \item lossless decoded fields;
  \item canonical identity and semantic payload hashes.
\end{itemize}

A source observation is never silently dropped.  It maps to one canonical event, a provenance group, or an explicit anomaly.

\subsection{Canonical event schema}

One canonical event row contains:

\begin{itemize}
  \item canonical event identity;
  \item full immutable log envelope;
  \item actual ABI-derived topic/event class;
  \item typed decoded fields;
  \item semantic payload hash;
  \item direct and group-level provenance references;
  \item duplicate/conflict classification;
  \item source coverage/range references.
\end{itemize}

\subsection{Condition-resolution relation}

One row per canonical condition-resolution event contains:

\begin{itemize}
  \item condition and question identifiers;
  \item oracle and outcome-slot count;
  \item ordered payout numerators as lossless integer strings;
  \item denominator and normalized class;
  \item protocol-final block, transaction, and log identity;
  \item upstream oracle and adapter linkage where exact;
  \item emission-order class;
  \item source/provenance references.
\end{itemize}

\subsection{Redemption relation}

One row per canonical redemption event contains:

\begin{itemize}
  \item condition identifier;
  \item redeemer and collateral token;
  \item parent collection identifier;
  \item ordered index sets;
  \item exact payout amount;
  \item block, transaction, and log identity;
  \item protocol-finality linkage;
  \item positive/zero payout indicator;
  \item source/provenance references.
\end{itemize}

\subsection{Condition analytical relation}

The condition-level analytical row records:

\begin{itemize}
  \item preparation and protocol-finality states;
  \item payout vector and class;
  \item exact or graded oracle/adapter clocks;
  \item first-any and first-positive redemption endpoints;
  \item censoring time and event indicators;
  \item ordered holder-state transitions;
  \item observed cumulative payout coordinates;
  \item metadata and route strata;
  \item gate states.
\end{itemize}

Addresses remain pseudonymous.  Row-level disclosure requires a review independent of analytical validity.

\section{Estimand, Denominator, and Gate Dictionary}
\label{app:estimands_gates}

\begin{longtable}{@{}L{0.20\textwidth}L{0.28\textwidth}L{0.42\textwidth}@{}}
\caption{Primary estimands and their populations.}\label{tab:estimand_dictionary}\\
\toprule
Estimand & Population / risk set & Interpretation\\
\midrule
\endfirsthead
\toprule
Estimand & Population / risk set & Interpretation\\
\midrule
\endhead
Question-to-condition coverage & frozen question cohort & exact mapping availability, not market completeness\\
Preparation incidence & frozen question cohort & observed condition registry coverage\\
Protocol-finality incidence & exact linked conditions & share with a valid resolution event by snapshot\\
\(L^{OP}\) & conditions with compatible oracle/protocol clocks & oracle-to-protocol mechanism interval\\
\(L^{OA}\) & conditions with compatible oracle/adapter clocks & oracle-to-adapter interval\\
\(D^{PA}\) & conditions with exact adapter/protocol events & signed emission gap plus order class\\
Payout-vector class & valid condition-resolution events & exact protocol payout rule class\\
\(L^{PR}\) & protocol-final conditions with follow-up & time to any observed redemption\\
\(L^{PR+}\) & protocol-final conditions with follow-up & time to first positive-payout redemption\\
Ordered transition probability & protocol-final conditions & no-redemption/zero-only/positive multistate history\\
Observed cumulative payout & canonical redemption events & realized collateral through observed actions\\
Entitlement completion & conditions with exact \(W_c\) & reported only if \(\valuegate\) passes\\
\bottomrule
\end{longtable}

\subsection{Gate consequence matrix}

\begin{longtable}{@{}L{0.16\textwidth}L{0.35\textwidth}L{0.39\textwidth}@{}}
\caption{Gate consequences.}\label{tab:gate_consequences}\\
\toprule
Gate & Pass requirement & Consequence of failure\\
\midrule
\endfirsthead
\toprule
Gate & Pass requirement & Consequence of failure\\
\midrule
\endhead
G-ABI & topic registry and decoder derived from pinned ABI & no target event classification\\
G-COVERAGE & zero-gap/no-unexplained-overlap ledger & no population-wide target-event rate\\
G-IDENTITY & exact transaction hash, log index, full envelope & no canonical multiplicity or order claim\\
G-DECODE & lossless typed decoding with raw envelope retained & no payout/event-field analysis\\
G-LINK & exact question-condition/upstream linkage & no question-level path claim\\
G-TIMING & compatible clock grade and order & no point-duration or survival claim\\
G-PRICE & valid executable price/quote observation & no cash-equivalent wedge claim\\
G-VALUE & exact balance/entitlement denominator & no entitlement-completion claim\\
G-REBUILD & deterministic network-disabled rebuild & no result based on unreproducible analytical outputs\\
\bottomrule
\end{longtable}

\subsection{Claim-strength labels}

Each analytical claim receives one label:

\begin{itemize}
  \item \textbf{Population exact.} Target coverage, identity, decoding, and denominator pass.
  \item \textbf{Audited subset exact.} Events and relations are exact within an explicitly incomplete coverage set.
  \item \textbf{Interval identified.} Endpoint lies between exact lower and upper bounds.
  \item \textbf{Observed holder action.} Exact event, but not universal holder state.
  \item \textbf{Proxy/model-based.} Depends on an explicit proxy or statistical model.
  \item \textbf{Unmeasured.} Evidence does not identify the coordinate.
  \item \textbf{Blocked.} A load-bearing gate fails.
\end{itemize}

A table can contain multiple labels when columns rely on different evidence grades.

\section{Survival and Multistate Protocol}
\label{app:survival}

\subsection{Time origin and censoring}

The time origin is the canonical protocol-finality event.  Conditions without an observed valid protocol-finality event do not enter the redemption risk set.  Follow-up ends at block \(90{,}114{,}204\), whose exact canonical timestamp is 2026-07-12T17:11:41Z.  Administrative censoring is deterministic given the frozen snapshot.

\subsection{Endpoints}

Endpoint A is first observed redemption of any payout amount.  Endpoint B is first observed strictly positive payout.  A prior zero-payout redemption ends Endpoint A but does not censor or preclude Endpoint B; follow-up for Endpoint B continues until a positive payout or the fixed snapshot.

\subsection{Kaplan--Meier outputs}

For each endpoint the release includes:

\begin{itemize}
  \item event and censoring counts;
  \item numbers at risk at prespecified horizons;
  \item Kaplan--Meier coordinates;
  \item median and selected quantiles where identifiable;
  \item restricted mean event-free time over a registered horizon;
  \item route, payout, adapter, and metadata strata.
\end{itemize}

Ties are handled at the available clock resolution.  Same-transaction order is preserved separately from the duration in seconds.

\subsection{Ordered multistate outputs}

The states are no redemption, zero-only observed, and positive observed.  Aalen--Johansen transition probabilities estimate the ordered transitions, including the direct no-redemption-to-positive path.  Zero-only is not treated as an absorbing competing cause for positive redemption.

\subsection{Interval-observed and unavailable clocks}

When protocol finality is exact and redemption is exact, the interval is point-observed.  When an upstream clock is interval-observed, oracle-to-protocol estimates are reported as intervals or excluded from point-duration tables.  Missing clocks are not midpoint-imputed in primary results.

\subsection{Inference boundary}

The cohort is a frozen historical population.  Nonparametric uncertainty intervals summarize the event-history estimator under stated censoring assumptions.  They do not identify causal effects or guarantee transfer to future conditions, other venues, or different fee/custody regimes.

\subsection{Required diagnostics}

The release includes:

\begin{enumerate}
  \item event/censoring balance by route and payout class;
  \item survival curves with numbers at risk;
  \item sensitivity to the positive endpoint definition;
  \item comparison with the incorrect market-creation time origin;
  \item zero-event and same-block timing diagnostics;
  \item common-support diagnostics for stratified comparisons.
\end{enumerate}

\section{Additional Proofs and Formal Audit}
\label{app:proofs}

\begin{lemma}[Protocol-finality ordering]
\label{lem:protocol_order}
Suppose the observed CTF contract accepts a redemption event only after payout numerators for its condition have been recorded.  Then every accepted redemption event satisfies
\[
  t_{cj}\geq\tresolve(c).
\]
\end{lemma}

\begin{proof}
By hypothesis, the redemption transition requires an already recorded payout rule.  The first valid event establishing that rule occurs at \(\tresolve(c)\).  A redemption accepted before that time would violate the transition precondition.  Hence the inequality holds.
\end{proof}

\begin{lemma}[Administrative censoring does not imply failure]
\label{lem:censoring_not_failure}
For a condition resolved before \(\tstar\), absence of a qualifying redemption event by \(\tstar\) identifies right censoring of the first-redemption endpoint, not an infinite redemption time.
\end{lemma}

\begin{proof}
The data contain no event after the fixed observation boundary.  A redemption may occur after \(\tstar\) without contradicting the observed history.  Therefore the event time is known only to exceed the available follow-up interval.
\end{proof}

\begin{proposition}[Risk-set origin matters]
\label{prop:risk_set_origin}
A survival curve measured from market creation generally differs from first-redemption survival measured from protocol finality, even when both use the same redemption events.
\end{proposition}

\begin{proof}
Let \(t_0(c)\) be market creation.  The two durations are \(\tredeem(c)-t_0(c)\) and \(\tredeem(c)-\tresolve(c)\).  Their difference is \(\tresolve(c)-t_0(c)\), which varies across conditions and includes time during which redemption is not yet admissible.  The distributions therefore need not coincide and answer different questions.
\end{proof}

\begin{proposition}[Positive-redemption follow-up continues after zero payout]
\label{prop:zero_not_terminal_positive}
A zero-payout redemption before \(\tstar\) does not identify failure of a later positive-payout redemption for the same condition.
\end{proposition}

\begin{proof}
A condition can have several position classes and holders.  A zero-payout event can burn a losing or otherwise zero-valued position while a winning position remains outstanding.  Another holder can later redeem for positive collateral.  Therefore the zero event does not preclude the positive endpoint.
\end{proof}

\begin{proposition}[Event coverage and entitlement coverage are logically independent]
\label{prop:event_value_independence}
Zero-gap coverage of all condition-resolution and redemption events does not imply complete entitlement coverage, and complete entitlement balances do not imply event-range coverage.
\end{proposition}

\begin{proof}
For the first direction, an event-complete ledger can omit unredeemed balances, so \(W_c\) remains unidentified.  For the second, a balance snapshot can identify outstanding positions while historical redemption events and their timing remain missing.  Neither evidence set logically contains the other.
\end{proof}

\begin{proposition}[Canonical multiplicity under repeated transport]
\label{prop:canonical_multiplicity}
Let \(S\geq1\) provider observations share one canonical identity and one semantic payload, and let exact chain evidence contain one corresponding event.  Then the canonical event count is one, the provenance count is \(S\), and the number of duplicate source appearances is \(S-1\).
\end{proposition}

\begin{proof}
The chain identity establishes one event.  Every provider occurrence remains a distinct observation/provenance relation.  Counting all occurrences as events would overstate chain multiplicity by \(S-1\); discarding them would lose provenance.  The stated accounting preserves both quantities.
\end{proof}

\subsection{Formal-claim boundary}

The formal results establish identification and accounting properties conditional on the pinned transition system.  They do not prove:

\begin{itemize}
  \item that every historical condition used one standard adapter branch;
  \item that holder redemption is economically optimal or immediate;
  \item that a particular leverage rule is welfare optimal;
  \item that protocol finality has the same legal effect in every jurisdiction;
  \item that address-level activity identifies beneficial owners.
\end{itemize}

Those questions require empirical, legal, or mechanism-design evidence beyond the propositions.

\subsection{Proof inventory}

The article contains direct proofs for:

\begin{enumerate}
  \item same-transaction order preservation;
  \item oracle/protocol non-equivalence;
  \item protocol/holder non-equivalence;
  \item scalar resolution timestamp insufficiency;
  \item payout-vector scale invariance;
  \item redemption-event insufficiency for entitlement completion;
  \item source-duplication accounting;
  \item protocol-finality ordering;
  \item censoring interpretation;
  \item risk-set origin;
  \item zero-to-positive endpoint logic;
  \item event/value coverage independence;
  \item canonical multiplicity.
\end{enumerate}

\section{Reproducibility and Evidence Lineage}
\label{app:reproducibility_lineage}

\subsection{Analytical artifacts}

The empirical tables and figures are derived from a fixed set of analytical relations:

\begin{enumerate}
  \item the question-to-condition mapping and attrition relation;
  \item protocol-finality paths and exact event-order keys;
  \item the ordered payout-vector taxonomy;
  \item first-redemption, first-positive-redemption, and multistate histories;
  \item observed payout-realization and gate-attrition tables; and
  \item source relations for figures, tables, and alternative text.
\end{enumerate}

These relations preserve denominator composition, clock definitions, and the evidence class of
each reported quantity.  Numerical values are generated from frozen inputs rather than entered
individually into tables.

\subsection{Event identity and coverage}

The target event family is derived from the pinned ABI rather than from stream labels.  A
canonical event is identified by
\[
  (\text{chain},\text{contract},\text{transaction hash},\text{log index}).
\]
The retained envelope includes block identity, timestamp, ordered topics, data, decoded fields,
and provenance.  A nullable transaction index is auxiliary; exact log index remains
load-bearing for event order.

Population-wide target-event claims require a zero-gap range ledger over the registered
deployment-to-snapshot interval.  Recursive ranges are inclusive, non-overlapping, and linked
to their parents.  A provider response is transport evidence rather than the event itself:
duplicate source observations remain provenance when they map to the same canonical event.

\subsection{Lineage and reproducibility}

The analytical lineage is
\[
\text{source observation}
\rightarrow \text{canonical event}
\rightarrow \text{exact relation}
\rightarrow \text{analytical table or figure}.
\]
For each reported table or figure, the lineage records its input relations, schema, key set,
denominator, clock origin, and content hash.  A network-disabled rebuild compares event identity,
semantic payloads, coverage, linkage, attrition, survival relations, and figure-source inputs.
Where deterministic file metadata permits, byte identity is checked; otherwise the check is
semantic.

\subsection{Evidence and disclosure boundary}

The scientific record distinguishes public aggregate tables, reconstruction relations needed to
reproduce the stated claims, and retained anomaly evidence.  Public blockchain visibility does
not by itself identify beneficial ownership, intent, institutional identity, or legal status.
Accordingly, reported tables aggregate redemption activity by condition, route, payout class, or
timing stratum unless row-level disclosure is necessary for the stated reproducibility purpose.

The data and code availability statement identifies the materials needed to assess the reported
claims.  Absence of an exact entitlement denominator remains a scientific limitation: observed
redemption events identify realized actions and paid collateral, not universal entitlement
completion.

\section{Version-Locked Hypotheses and Interpretation Boundaries}
\label{app:hypotheses_disposition}

\subsection{Hypothesis disposition rule}

The seven hypotheses in \Cref{sec:empirical_contract} are retained whether they pass, fail, or
remain blocked.  For each hypothesis, the analysis reports its statement, required population
and gates, result or uncertainty summary, disposition, and permitted claim language.  A failed
threshold is not replaced after inspection, and a blocked hypothesis is not evaluated on an
opportunistic denominator.

\subsection{Interpretation conditions}

The formal units, estimands, time origins, event taxonomy, and entitlement boundary are fixed by
the analytical design.  Population-wide interpretations require the exact target-event dataset to
satisfy semantic and zero-gap coverage requirements.  Timing and survival interpretations require
the relevant compatible clocks and risk sets.  Reported tables and figures are generated from
frozen analytical relations and assessed by the reproducibility checks described in
\Cref{app:reproducibility_lineage}.

The absence of an exact balance-consistent entitlement denominator leaves G-VALUE blocked.  It
does not prevent reporting protocol-finality, payout-vector, censoring, or observed-redemption
quantities within their stated populations and evidence classes.

\section*{Data and Code Availability}
The article reports aggregate lifecycle, payout, redemption, and censoring results; holder-level entitlement reconstruction and raw holder-level records are not distributed.

\section*{Generative AI Disclosure}
OpenAI ChatGPT and Codex were used for editorial and technical assistance during manuscript preparation. The author made all substantive research decisions, reviewed the final manuscript, and assumes full responsibility for its contents.

\section*{Funding}
This research received no external funding.

\section*{Competing Interests}
The author is affiliated with the Research Department of Devnull FZCO and leads the ForesightFlow research programme. No external sponsor influenced the research design, analysis, interpretation, or decision to publish. The article does not evaluate a commercial product or make investment recommendations.

\section*{Acknowledgments}
This work was completed within the ForesightFlow research programme. All remaining errors are the author's responsibility.

\printbibliography[heading=bibintoc,title={References}]
\end{document}